\documentclass[preprint,12pt]{elsarticle}
\usepackage{rotating}
\usepackage{amsfonts}
\usepackage{nicefrac}
\usepackage{float}
\usepackage{wrapfig}
\usepackage{amscd}
\usepackage{paralist}
\usepackage{verbatim}
\usepackage{tikz}
\usetikzlibrary{positioning}
\usetikzlibrary{graphs}
\usepackage{subcaption}
\usepackage{adjustbox}
\usepackage{algpseudocode}
\usepackage[ruled,linesnumbered]{algorithm2e}
\usepackage{setspace}
\usepackage{listings}
\usepackage{pgf}
\usepackage{graphicx}
\usepackage{dashbox}
\usepackage{float}
\usepackage{syntax}
\usepackage{hyperref}
\usepackage{hhline, booktabs}
\usepackage{array,multirow}	
\usepackage{tablefootnote}
\usepackage{amsthm}
\newtheorem{theorem}{Theorem}
\usepackage{rotating}
\usepackage{amsthm}
\usepackage{amsfonts}
\usepackage{nicefrac}
\usepackage{wrapfig}
\usepackage{amscd}
\usepackage{verbatim}
\usepackage{tikz}
\usepackage{paralist}
\usepackage{microtype}
\usepackage{makecell}
\usepackage{listings}
\usepackage{caption}
\usepackage{xcolor}
\usepackage{color,soul}
\usepackage{amssymb}

\usepackage{graphicx}
\usepackage{amsmath}
\usepackage{thmtools}

\usepackage{listings}
\definecolor{codegreen}{rgb}{0,0.6,0}
\definecolor{codegray}{rgb}{0.5,0.5,0.5}
\definecolor{codepurple}{rgb}{0.58,0,0.82}
\definecolor{backcolour}{rgb}{0.95,0.95,0.92}

\usepackage{xcolor}

\newcommand{\para}[1]{\smallskip\noindent\textbf{#1.}}

\renewcommand\thmcontinues[1]{Continued}

\newcommand{\slap}{\texttt{SPL}\xspace}

\newcommand{\sloop}{\circledast}

\newcommand{\cfg}[1]{\textsf{cfg}(#1)}
\newcommand{\lt}[1]{\textsf{lt}(#1)}
\newcommand{\opt}[1]{\textsc{Opt}[#1]}
\newcommand{\vars}{\mathbb{V}}
\newcommand{\cost}{\textsc{Cost}}
\newcommand{\DP}{\texttt{dp}}
\newcommand{\compatible}{\leftrightharpoons}

\newcommand{\oseries}[2]{#1 \otimes #2}
\newcommand{\oparal}[2]{#1 \oplus  #2}
\newcommand{\oloop}[1]{#1^\circledast}
\usepackage{amsmath}

\DeclareMathOperator*{\argmin}{arg\,min}

\tikzset{
    arrow/.style={thick,->,>=stealth, draw=blue!20!black},
    vertex/.style={rectangle, minimum width=10mm, minimum height=10mm, fill=green!10, draw=black!50},
    whitevertex/.style={rectangle, minimum width=10mm, minimum height=10mm, fill=white, draw=black!50},
    grayvertex/.style={rectangle, minimum width=10mm, minimum height=10mm, fill=gray!20, draw=gray!80},
    redvertex/.style={rectangle, minimum width=10mm, minimum height=10mm, fill=orange!20, draw=orange!80},
    bluevertex/.style={rectangle, minimum width=10mm, minimum height=10mm, fill=green!30, draw=green!80},
	vertexs/.style={thick, rectangle, minimum size=10mm, fill=blue!20, draw=black},
    vertext/.style={thick, rectangle, minimum size=10mm, fill=red!20, draw=black},
    vertexc/.style={thick, rectangle, minimum size=10mm, fill=green!10, draw=black},
    vertexb/.style={thick, rectangle, minimum size=10mm, fill=black, draw=black, text=white},
    vertexla/.style={minimum size=5mm, draw=none, text=red},
    emptynode/.style={thick, rectangle, minimum size=10mm, fill=white, draw=white, text=white},
    rect/.style={thick, rectangle, minimum size=10mm, fill=black!5, text=black}
}

\usepackage{amssymb}
\usepackage{amsmath}
\journal{Journal of Systems Architecture}

\begin{document}

\begin{frontmatter}



\title{Series-Parallel-Loop Decompositions\\of Control-flow Graphs} %% Article title


\author[oxford]{Xuran Cai} %% Author name
\author[oxford]{Amir Kafshdar Goharshady} %% Author name
\author[hkust]{S Hitarth} %% Author name
\author[hkust]{Chun Kit Lam} %% Author name

\affiliation[oxford]{organization={Department of Computer Science, University of Oxford},%Department and Organization
            	addressline={Wolfson Building, Parks Road}, 
            	city={Oxford},
            	country={United Kingdom}}

\affiliation[hkust]{organization={Department of Computer Science and Engineering, Hong Kong University of Science and Technology},%Department and Organization
            addressline={Clear Water Bay}, 
            city={New Territories},
            country={Hong Kong},
          }

%% Abstract
\begin{abstract}
Control-flow graphs (CFGs) of structured programs are well known to exhibit strong sparsity properties. Traditionally, this sparsity has been modeled using graph parameters such as treewidth and pathwidth, enabling the development of faster parameterized algorithms for tasks in compiler optimization, model checking, and program analysis. However, these parameters only approximate the structural constraints of CFGs: although every structured CFG has treewidth at most~7, many graphs with treewidth at most~7 cannot arise as CFGs. As a result, existing parameterized techniques are optimized for a substantially broader class of graphs than those encountered in practice.

In this work, we introduce a new grammar-based decomposition framework that characterizes \emph{exactly} the class of control-flow graphs generated by structured programs. Our decomposition is intuitive, mirrors the syntactic structure of programs, and remains fully compatible with the dynamic-programming paradigm of treewidth-based methods. Using this framework, we design improved algorithms for two classical compiler optimization problems: \emph{Register Allocation} and \emph{Lifetime-Optimal Speculative Partial Redundancy Elimination (LOSPRE)}. Extensive experimental evaluation demonstrates significant performance improvements over previous state-of-the-art approaches, highlighting the benefits of using decompositions tailored specifically to CFGs.
\end{abstract}


%% Keywords

\end{frontmatter}

\section{Introduction and Related Works} \label{sec:intro}
Many fundamental tasks in program analysis, compiler optimization, and formal
verification are traditionally approached by reducing them to graph-theoretic
problems. Classical and widely-studied examples include register allocation%
~\cite{DBLP:conf/pldi/Chaitin82,DBLP:conf/lcpc/BouchezDGR06}, LOSPRE%
~\cite{DBLP:conf/scopes/Krause21}, cache optimization techniques such as data
packing~\cite{thabit1982cache} and cache-conscious data placement%
~\cite{DBLP:conf/asplos/CalderKJA98,DBLP:conf/iwmm/BegB10}, interprocedural
data-flow analysis~\cite{DBLP:conf/popl/RepsHS95}, algebraic program
analysis~\cite{DBLP:conf/cav/KincaidRC21}, equality saturation and extraction%
~\cite{egg-equality-saturation,treewidth-equality-extraction}, and
$\mu$-calculus model checking~\cite{DBLP:conf/cav/Obdrzalek03}, among many
others.

The graph problems arising from these tasks are frequently NP-hard%
~\cite{DBLP:conf/pldi/Chaitin82,thabit1982cache,DBLP:conf/popl/Lavaee16}, and in
some cases even hard to approximate unless $\mathrm{P}=\mathrm{NP}$%
~\cite{DBLP:conf/popl/PetrankR02,DBLP:journals/njc/PetrankR05}. This naturally
raises the question of whether solving the most general form of these graph
problems---that is, considering \emph{all} graphs as possible input instances---is
an appropriate model for the original program analysis or compiler optimization
task. In practice, the graphs underlying these tasks are almost always the
control-flow graphs (CFGs) of well-structured programs, or graphs closely
derived from them (e.g., by duplicating vertices). Such CFGs are highly sparse,
and many graphs cannot arise as the CFG of any real program. Hence, algorithms
designed for unrestricted graph classes may be solving a strictly harder problem
than necessary.

A substantial body of work has therefore focused on formalizing the sparsity
and structural properties of CFGs and related graph representations. For
example, program dependence graphs (PDGs)~\cite{DBLP:journals/toplas/FerranteOW87}
capture both data and control dependencies, enabling more precise optimization
analyses. More recently, Generalized Points-to Graphs (GPGs)%
~\cite{DBLP:journals/toplas/GharatKM20} have been proposed for pointer analysis,
retaining only those parts of the CFG relevant to aliasing. A particularly
influential result in understanding CFG sparsity is due to Thorup%
~\cite{THORUP1998159}, who proved that CFGs of structured (i.e., goto-free)
programs in many widely-used languages have treewidth at most~$6$. Treewidth%
~\cite{DBLP:journals/jct/RobertsonS84} is a classical measure of how
tree-like a graph is: a graph of treewidth $k$ can be decomposed into small
overlapping vertex sets (``bags'') of size at most $k+1$, arranged in a
tree-shaped structure. Figure~\ref{fig:tw} illustrates such a decomposition,
and further background can be found in%
~\cite{DBLP:conf/sofsem/Bodlaender05,cygan2015parameterized}. This structure
enables efficient bottom-up and top-down dynamic programming techniques, as well
as divide-and-conquer strategies%
~\cite{DBLP:journals/jal/RobertsonS86,DBLP:conf/icalp/Bodlaender88,DBLP:journals/dam/GoharshadyHM16,DBLP:journals/ress/GoharshadyM20,DBLP:conf/blockchain2/MeybodiGHS22,DBLP:conf/fsttcs/AhmadiCGMSZ22,DBLP:conf/blockchain2/BarakbayevaFGGN24,DBLP:journals/jco/MeybodiGHS25}.

\begin{figure}
	\centering
		$\begin{matrix}
		\hspace{-.7cm}
		\begin{tikzpicture}[node distance={15mm}, line width=1pt, main/.style = {draw, circle}] 
			\node[main] (1) at (0,0){$1$}; 
			\node[main] (2) [below of=1]{$2$}; 
			\node[main] (3) [below left of=2]{$3$}; 
			\node[main] (4) [below right of=2]{$4$}; 
			\node[main] (5) [below left of=3]{$5$}; 
			\node[main] (6) [below right of=4]{$6$}; 
			\draw [] (1) -- (2); 
			\draw [] (2) -- (3); 
			\draw [] (2) -- (4); 
			\draw [] (3) -- (4); 
			\draw [] (3) -- (5); 
			\draw [] (3) -- (6); 
			\draw [] (4) -- (6); 
			
			\node[text width=0cm, red] at (2,0){$ $};
			\draw[rotate=0,line width=1pt,dashed,cyan!60!blue] (0,-0.7) ellipse (20pt and 37pt);
			
			\node[text width=0cm, red] at (4.8,-1.8){$ $};
			\draw[rotate=0,line width=1pt,dashed,red] (0,-2.2) ellipse (50pt and 35pt);
			
			\node[text width=0cm, red] at (5.8,-3.6){$ $};
			\draw[rotate=-15,line width=1pt,dashed,green!60!black] (1.3,-2.8) ellipse (65pt and 30pt);
			
			\node[text width=0cm, red] at (-0.1,-3.6){$ $};
			\draw[rotate=-45,line width=1pt,dashed,gray] (1,-3.3) ellipse (20pt and 37pt);
		\end{tikzpicture}
		&   \hspace{-3.2cm}
		\begin{tikzpicture}[node distance={20mm}, line width=1pt, main/.style = {draw, rectangle}] 
			\node[main] (1) [red]{$\{2, 3, 4\}$}; 
			\node[main] (2) [cyan!60!blue, below right of=1]{$\{1, 2\}$}; 
			\node[main] (3) [below of=1,node distance=21mm, gray]{$\{3, 5\}$}; 
			\node[main] (4) [green!60!black, below left of=1]{$\{3, 4, 6\}$}; 
			\draw [] (1) -- (2); 
			\draw [] (1) -- (3);
			\draw [] (1) -- (4);
		\end{tikzpicture}
	\end{matrix}$
	\caption{A graph $G$ (left) and a tree decomposition of $G$ (right). This example is taken from~\cite{DBLP:journals/pacmpl/ConradoGL23}.}
	\label{fig:tw}
\end{figure}
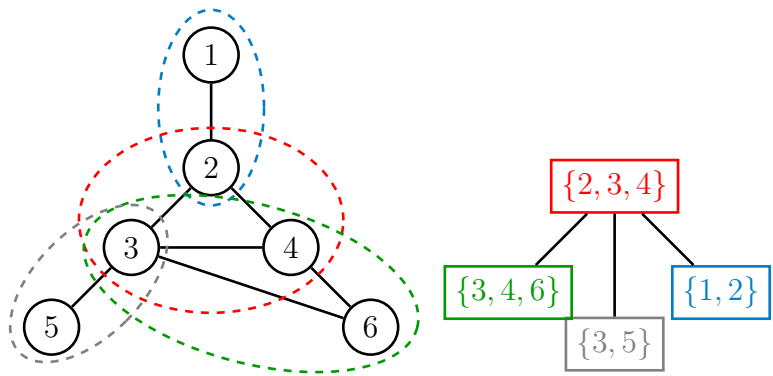

The result of~\cite{THORUP1998159} showed that the treewidth of CFGs is at most $6$ in languages such as \texttt{C} and \texttt{Pascal} and was followed by a number of similar results extending it to other languages, such as \texttt{Java}~\cite{DBLP:conf/alenex/GustedtMT02}, \texttt{Ada}~\cite{DBLP:conf/adaEurope/BurgstallerBS04} and \texttt{Solidity}~\cite{DBLP:conf/sac/ChatterjeeGG19}, both theoretically and experimentally. The work~\cite{DBLP:journals/dam/KrauseLS20} showed that the treewidth can be $7$ in some variants of \texttt{C}.

These bounded-treewidth results effectively opened up a completely new direction of research. When a program-related task, such as those in compiler optimization or verification, is reduced to a graph problem, we often need not solve it over general graphs, but only graphs of bounded treewidth, since they already cover the set of all possible CFGs. Thus, we are no longer looking for polynomial-time algorithms with a runtime of $O(n^c)$ where $n$ is the size of the input, but instead aim to find so-called \emph{piecewise-polynomial (XP)} or \emph{fixed-parameter tractable (FPT)} algorithms~\cite{downey2012parameterized} with  runtimes of the form $O(n^{f(k)})$ and $O(n^c \cdot f(k)),$ respectively, where the parameter $k$ is the treewidth and $f$ is an arbitrary computable function, possibly not a polynomial. Intuitively, since we know that the treewidth is small and at most $7$ in all the instances that we care about, i.e. CFGs, an algorithm with a runtime of $O(n^c \cdot f(k))$ is effectively polynomial-time for all real-world intents and purposes. For example, an algorithm that takes $O(n \cdot 2^k)$ will behave just like a linear-time algorithm over CFGs since $k \leq 7$ and thus $2^k$ is a constant.

This approach has been extremely successful in finding polynomial-time (FPT) algorithms for NP-hard graph problems arising in program-related tasks, making their real-world instances tractable. Some prominent examples are as follows:
\begin{itemize}
	\item After proving the treewidth bound,~\cite{THORUP1998159} showed that Chaitin-like register allocation, i.e.~assigning program variables to a limited number of registers so that no two variables with intersecting lifetimes (without allowing lifetime splitting) get assigned to the same register, can be approximated over graphs of bounded treewidth $k$ within a factor of $\lfloor k/2 + 1 \rfloor$ from optimal in linear time. This leads to a $4$-approximation for CFGs. In contrast, it is well-known that the problem is equivalent to graph coloring~\cite{DBLP:conf/pldi/Chaitin82} and thus hard-to-approximate within any constant factor over general graphs unless P=NP. 
	\item This was later improved to an optimal linear-time FPT algorithm for the decision variant of register allocation, i.e.~answering whether it is possible to have no spilling with a fixed number $r$ of registers~\cite{DBLP:conf/soda/BodlaenderGT98}.
	\item It was later shown that the problem of optimal-cost register allocation, i.e.~minimizing the total cost of spills given a fixed number of registers, is also solvable in polynomial time (XP) when the treewidth is bounded~\cite{DBLP:conf/cc/Krause13}. 
	\item In~\cite{DBLP:journals/iandc/Courcelle90,DBLP:journals/algorithmica/BoriePT92}, it was established that any formula in the monadic second-order logic of graphs can be model-checked in linear time over graphs of bounded treewidth. This was extended in~\cite{DBLP:conf/lpar/ChimesIZ24} to embeddable bounded treewidth a more general family of graph grammars.
	\item \cite{DBLP:journals/pacmpl/ChatterjeeGOP19,DBLP:conf/pldi/AhmadiDGP22} showed that both data packing and cache-conscious data-placement, which are two classical cache management problems in compiler optimization, are solvable/approximable in polynomial time over instances with bounded treewidth, despite their hardness over general instances~\cite{DBLP:conf/popl/PetrankR02,DBLP:conf/popl/Lavaee16}.
	\item The work~\cite{DBLP:conf/cav/Obdrzalek03} showed that the winner in a parity game whose underlying graph has bounded treewidth can be decided in polynomial time. The well-known correspondence between parity games and $\mu$-calculus/LTL model checking makes this applicable to real-world instances obtained from CFGs.
\end{itemize}

The bounded treewidth result has been useful even when the original graph problem already admitted a polynomial-time solution on all graphs. In many problems, one can obtain a lower-degree polynomial algorithm by exploiting the treewidth. For example:
\begin{itemize}
	\item The works~\cite{DBLP:conf/popl/ChatterjeeGIP16,DBLP:conf/esop/ChatterjeeGIP20,DBLP:conf/atva/ChatterjeeGP17,DBLP:journals/toplas/ChatterjeeGGIP19,DBLP:conf/esop/ChatterjeeGIP20} consider interprocedural data-flow analyses, i.e.~problems such as null-pointer identification and available expressions, and provide linear-time FPT algorithms for their on-demand variant. This is in contrast to the non-parameterized methods' quadratic runtime~\cite{DBLP:conf/popl/RepsHS95,DBLP:conf/sigsoft/HorwitzRS95}. A similar result is obtained for algebraic program analysis in~\cite{DBLP:journals/pacmpl/ConradoGKTZ23}.
	\item Various problems in linear algebra~\cite{fomin2018fully}, including Gaussian elimination and many classical qualitative and quantitative tasks on Markov chains and Markov decision processes~\cite{DBLP:conf/cav/ChatterjeeL13,DBLP:conf/atva/AsadiCGMP20} can be solved in linear-time FPT if the underlying graph has bounded treewidth. This is in contrast to the best-known algorithms for general graphs which take $O(n^\omega)$ where $\omega$ is the matrix multiplication constant. A similar improvement was recently obtained for linear programming~\cite{DBLP:conf/stoc/DongLY21}.
\end{itemize}

All the advances above in program-related graph problems are based on the same fundamental intuition: CFGs are sparse and solving problems over CFGs is very different and often much easier than solving the same problems over general graphs. Thus, taking this intuition to its natural conclusion, one should wonder if bounded treewidth sufficiently captures the sparsity of control-flow graphs.

On the one hand, it is easy to come up with graphs of bounded treewidth that are not realizable as a CFG. For example, consider a graph with $n$ connected components, each of which is a complete graph on $6$ vertices, i.e. $K_6.$ Such a graph will have a treewidth of $5$ but is clearly not the CFG of any structured program, given that any node in a CFG can have an out-degree of at most two. On the other hand, the recent work~\cite{DBLP:journals/pacmpl/ConradoGL23} shows that CFGs not only have bounded treewidth but also often have small pathwidth, i.e.~they can be decomposed into small bags that are connected to each other not just in a tree-like manner but in a path-like manner. It then shows that this enables much more efficient dynamic programming algorithms and significant runtime gains for problems such as register allocation. Intuitively, this is because dynamic programming on paths is often simpler than on trees. However, the fundamental question remains: Is bounded pathwidth exactly capturing the sparsity of control-flow graphs? The answer is negative with the same counterexample as above. Thus, we consider this problem: Can we come up with a standard way of decomposing graphs that captures exactly the set of CFGs? Is there a decomposition method that (i)~enables fast dynamic programming for problems in compiler optimization and formal verification, and (ii)~is applicable to precisely the same set of graphs as CFGs?

In this work, we present such a decomposition based on a graph grammar that closely mimics the grammars used for defining the syntax of structured programming languages. Our approach is similar both to the standard definitions of series-parallel graphs~\cite{takamizawa1982linear} and previous program analysis methods based on graph grammars~\cite{wirth1974composition,kennedy1977applications}. However, our graphs are more general in the sense that they cover precisely the set of all control-flow graphs, including CFGs of structured programs that contain \texttt{break} and \texttt{continue} statements. We define a natural decomposition of CFGs based on our graph grammar and show that it can be used for efficient dynamic programming over the CFGs. As two concrete use-cases, we consider the Chaitin-like Register Allocation  problem and Lifetime-Optimal Speculative Partial Redundancy Elimination (LOSPRE), which are both classical and well-studied in the literature with both parameterized and non-parameterized solutions. We provide algorithms based on our decomposition that are faster than not only the previous non-parameterized solutions, but even the state-of-the-art approaches that exploited bounded treewidth.

Finally, we present comprehensive experimental results in Section~\ref{sec:exp}, demonstrating that our approach achieves significant runtime improvements in practice. Notably, for spill-free register allocation, our method is the first exact algorithm capable of handling real-world instances with up to 20 registers, representing a substantial advancement over previous methods that were limited to just 8 registers. This is particularly important as standard architectures, such as the \texttt{x86} family, typically feature 16 registers. Consequently, our algorithm is the first exact solution applicable to these architectures, eliminating the need for approximate or heuristic methods for spill-free register allocation. For LOSPRE, our approach shows an average improvement of approximately five times compared to the previous state-of-the-art treewidth-based method. Given that this earlier solution was already highly optimized, this represents a significant enhancement.

While runtime improvements for Chaitin-like register allocation and LOSPRE are significant in their own right, we believe that similar improvements can be obtained for a much wider family of program-related analyses by relying on our notion of series-parallel-loop decomposition instead of treewidth/pathwidth.

In summary, this work takes the idea of exploiting the sparsity of control-flow graphs for faster graph algorithms to its ultimate end and defines a decomposition notion that, unlike parameters such as treewidth, captures \emph{precisely} the set of graphs that can arise as CFGs of structured programs. It then provides faster algorithms using this new notion of decomposition for two classical problems in compiler optimization, i.e.~the minimum-cost register allocation problem and LOSPRE. We expect that this kind of decomposition would also be useful for the many other problems in which parameterizations by treewidth/pathwidth were applied in the past. 
\section{Background}

In this section, we review the foundational concepts that underpin the algorithmic framework developed later in this work. We begin with the theory of parameterized
algorithms, which offers tools for exploiting structural properties of hard problems, followed by a discussion of tree decompositions and treewidth, two of the most
influential parameters used to capture graph sparsity. These notions play a crucial role in many compiler-optimization tasks, particularly those involving control-flow graphs (CFGs). 

\subsection{Parameterized Algorithms}

Parameterized algorithms~\cite{cygan2015parameterized} provide a refined framework for analyzing computational problems by isolating one or more parameters that capture structural aspects of the input. Instead of measuring complexity solely in terms of input size \(n\), parameterized complexity introduces an auxiliary parameter \(k\), allowing the development of algorithms that remain efficient even for NP-hard problems when \(k\) is relatively small. This paradigm has led to breakthroughs across theoretical computer science, affecting areas such as graph algorithms, formal verification, and compiler optimization~\cite{gpar1, gpar2}.

A central notion in the field is \emph{Fixed-Parameter Tractability} (FPT)~\cite{downey2013fundamentals}. A problem is fixed-parameter tractable if it can be solved in time
\[
O(f(k)\cdot n^{c}),
\]
where \(f\) is a computable function depending solely on the parameter \(k\), and \(c\) is a constant independent of both \(n\) and \(k\).  
While \(f(k)\) is often exponential, the key advantage is that the combinatorial explosion is confined to a small part of the input, making the algorithm efficient for small values of \(k\).

A canonical example is the \emph{Vertex Cover} problem~\cite{khuller2002algorithms}. Given a graph \(G=(V,E)\), the objective is to find a minimum-size subset of vertices that contains at least one endpoint of every edge. Although the general problem is NP-hard, the \emph{parameterized} question,
\[
\text{Is there a vertex cover of size at most } k?
\]
admits an elegant FPT algorithm. The branching strategy shown below demonstrates this idea: each step chooses an uncovered edge \((u,v)\) and branches on including either \(u\) or \(v\) in the solution.

\lstset{language=python}
\begin{lstlisting}

def vertex_cover(graph, k):
    if k < 0:
        return None
    if graph.isEmpty():
        return set()

    u, v = graph.edges[0]
    cover_with_u = vertex_cover(graph.remove_vertex(u), k - 1)
    cover_with_v = vertex_cover(graph.remove_vertex(v), k - 1)

    if cover_with_u is not None:
        return cover_with_u.union({u})
    if cover_with_v is not None:
        return cover_with_v.union({v})
    return None
\end{lstlisting}

This recursive process produces at most \(2^k\) branches, and each branch performs polynomial-time work, yielding a running time of \(O(2^k\cdot n)\). Thus, \(f(k)=2^k\), placing the problem  within FPT.

Beyond FPT, parameterized complexity also studies the broader class XP~\cite{downey2013fundamentals}. A problem is in XP if it has an algorithm that runs in time
\[
O(n^{f(k)}),
\]
meaning that for each fixed parameter value \(k\), the problem becomes polynomial-time solvable—though the degree of the polynomial may depend on \(k\). XP algorithms are particularly useful for moderate parameter values and serve as a bridge between classical polynomial-time algorithms and strictly fixed-parameter tractable ones.

As an example, consider the naive enumeration algorithm for the size-\(k\) Vertex Cover problem. One can simply enumerate all \(\binom{n}{k}\) subsets, checking each candidate in linear time. This approach runs in
\[
O(n^{k+1}),
\]
which fits the XP definition with \(f(k)=k+1\). Although this strategy is far less efficient than the FPT branching algorithm, it provides a correct and conceptually simple XP solution.

Parameterized reasoning has also become an important tool in practical domains such as compiler optimization and model checking. Many analyses over control-flow graphs benefit from small structural parameters (e.g., treewidth or pathwidth), enabling algorithms that operate in linear time when these parameters are bounded. This advantage has motivated a growing body of research applying parameterized techniques to optimize real-world programming languages and verification pipelines.

\subsection{Tree Decompositions and Treewidth}

Tree decompositions, introduced by Robertson and Seymour~\cite{robertson1986graph}, offer a powerful way to understand how ``tree-like'' a graph is. A tree decomposition of a graph \(G=(V,E)\) consists of a tree \(T\) along with a family of vertex subsets (bags) \(\{B_t\}_{t\in T}\). An example is shown in Figure~\ref{fig:tw}. This example originally appeared in~\cite{DBLP:journals/pacmpl/ConradoGL23}. The decomposition must satisfy the following conditions:

\begin{itemize}
    \item \textbf{Covering Condition:} Every vertex \(v\in V\) appears in at least one bag \(B_t\).
    \item \textbf{Edge Condition:} For every edge \((u,v)\in E\), there exists a bag that contains both \(u\) and \(v\).
    \item \textbf{Connectivity Condition:} For each vertex \(v\in V\), all nodes \(t\) such that \(v\in B_t\) must form a connected subtree of \(T\).
\end{itemize}

The width of a decomposition is the size of its largest bag minus one, and the \emph{treewidth} of \(G\) is the minimum width over all possible decompositions. Graphs with small treewidth behave much like trees, enabling a variety of dynamic programming techniques that would be infeasible on arbitrary graphs.

Tree decompositions have proved especially effective for solving special cases of NP-hard problems. Once the graph is decomposed into small, overlapping pieces, one can perform dynamic programming over the tree structure. Because each bag is bounded in size, the amount of computation per node is controlled, and the algorithm’s overall runtime often becomes linear in \(|V|\), with an exponential dependence only on the treewidth.

A well-known application appears in compiler analysis. Recent theoretical results demonstrate that every goto-free structured control-flow graph (CFG) has treewidth at most 7~\cite{THORUP1998159}. This allows treewidth to be treated as a fixed constant for many program analyses. Consequently, numerous optimization problems on CFGs—such as register allocation, data-flow analysis, and code motion—can be solved efficiently by exploiting the bounded-treewidth structure~\cite{gtw1, gtw2}.

Dynamic programming on a tree decomposition proceeds as follows: Let \(P\) be a problem defined on a graph \(G\), and let \((T,\{B_t\}_{t\in T})\) be a tree decomposition. For each node \(t\in T\), let \(G_t\) denote the induced subgraph on the bag \(B_t\). Suppose we construct a table \(S_t\) at each node satisfying:

\begin{enumerate}
    \item \(S_t\) contains sufficient information to evaluate problem \(P\) on \(G_t\).
    \item For leaf nodes, \(S_t\) is computed using only \(G_t\).
    \item For internal nodes, \(S_t\) is computed by combining the tables of its children along with the structure of \(G_t\).
\end{enumerate}

When these conditions hold, dynamic programming can be carried out bottom-up over the tree. Because each bag has size at most \(w+1\), where \(w\) is the treewidth, the computation per bag is bounded by a function of \(w\) alone. As a result, the overall runtime is \(O(f(w)\cdot n)\), which is linear in the size of the program when \(w\) is a small constant. This principle underlies many modern parameterized algorithms for compiler optimization and static analysis.

\section{SPL Decomposition} \label{sec:grammar}

To define structured programs, we follow a syntax similar to that of~\cite{THORUP1998159}. A program is generated from the following grammar:
\begin{equation} \label{gram:prog}
	\begin{array}{rl}
P := & \epsilon \mid \texttt{break} \mid \texttt{continue} \mid P ; P\\ & \mid \texttt{if} ~\varphi~ \texttt{then} ~P~ \texttt{else} ~P~ \texttt{fi} \mid \texttt{while} ~\varphi~ \texttt{do} ~P~ \texttt{od}.
\end{array}
\end{equation}
Here, $\epsilon$ is a neutral statement that does not affect control-flow, e.g.~a variable assignment, and $\varphi$ is a boolean expression. We say a program $P$ is \emph{closed} if every \texttt{break} and \texttt{continue} statement appears within the body of a \texttt{while} loop. The semantics of a program $P$ will be defined in the usual manner. In this section, we are only concerned with the control-flow graph of a program $P.$ We also note that other common constructs such as \texttt{for} loops and \texttt{switch} statements can be defined as syntactic sugar~\cite{THORUP1998159}. Specifically, a \texttt{switch} statement with $k$ jumps can be modeled as $k$ \texttt{if} statements.

Inspired by series-parallel graphs~\cite{DBLP:conf/stoc/ValdesTL79} and the approach of~\cite{kennedy1977applications}, we define a graph grammar that builds up new graphs by combining smaller graphs generated in the same grammar. To capture all control-flow graphs, we consider three composition operations: series, parallel and loop. Thus, we call the graphs generated by this grammar $\slap$ graphs. Also, unlike series-parallel graphs which have two distinguished terminals, i.e.~start and end, our graphs have four distinguished vertices, namely start ($S$), terminate ($T$), break ($B$) and continue ($C$).  
As base cases, the three graphs of Figure~\ref{fig:atom} are $\slap$ and called \emph{atomic} graphs $A_\epsilon, A_{\texttt{break}}$ and $A_{\texttt{continue}}$.

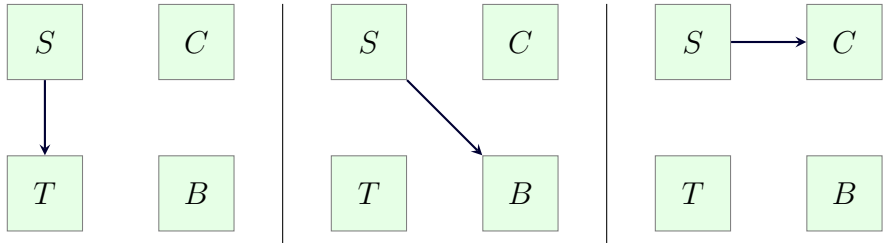
\begin{figure}[H]
	\begin{center}
	\begin{tabular}{cc|ccc|cc}
		\begin{tikzpicture}[scale=0.6]
			\node[vertex] (S) {$S$};
			\node[vertex] (T)  [below=of S] {$T$};
			\node[vertex] (C)  [right=of S] {$C$};
			\node[vertex] (B)  [below=of C] {$B$};
			\draw[arrow] (S) -- (T) ;
		\end{tikzpicture} & \quad & \quad &
		\begin{tikzpicture}[scale=0.6]
			\node[vertex] (S) {$S$};
			\node[vertex] (T)  [below=of S] {$T$};
			\node[vertex] (C)  [right=of S] {$C$};
			\node[vertex] (B)  [below=of C] {$B$};
			\draw[arrow] (S) -- (B) node [midway, above] {};
		\end{tikzpicture} & \quad & \quad &
		\begin{tikzpicture}[scale=0.6]
			\node[vertex] (S) {$S$};
			\node[vertex] (T)  [below=of S] {$T$};
			\node[vertex] (C)  [right=of S] {$C$};
			\node[vertex] (B)  [below=of C] {$B$};
			\draw[arrow] (S) -- (C) node [midway, above] {};
		\end{tikzpicture} 
	\end{tabular}
\end{center}
\caption{Atomic $\slap$ graphs: $A_\epsilon$ (left), $A_{\texttt{break}}$ (middle), and $A_{\texttt{continue}}$ (right).}
\label{fig:atom}
\end{figure}

\slap graphs are generated by the repeated application of series, parallel and loop operations that are denoted by $\oseries{}{}, \oparal{}{},$ and $\sloop$ respectively. Formally, the set of $\slap$ graphs is defined by the following context-free grammar: 
\begin{equation} \label{gram:graph}
G = A_\epsilon \mid A_{\texttt{break}} \mid A_{\texttt{continue}} \mid \oseries{G}{G} \mid \oparal{G}{G} \mid \oloop{G}.
\end{equation}
We denote an $\slap$ graph by the tuple $G = (V, E, S, T, B, C),$ where $V$ is the finite set of vertices of the graph, $E$ is the finite set of its edges, and $S, T, B, C \in V$ are the distinguished start, terminate, break and continue vertices. We sometimes drop $V$ and $E$ when they are clear from the context.

We now define our three operations. Let $G_1 = (V_1, E_1, S_1, T_1,$ $B_1, C_1)$ and $G_2 = (V_2, E_2, S_2, T_2, B_2, C_2)$ be two disjoint \slap graphs.

\begin{enumerate}
    \item \emph{Series Operation.}
    The graph $\oseries{G_1}{G_2}$ is constructed by first taking the disjoint union of
    $G_1$ and $G_2$ and then {merging specific pairs of vertices}.
    In particular, we merge and identify
    \[
        T_1 \equiv S_2,\qquad B_1 \equiv B_2,\qquad C_1 \equiv C_2,
    \]
    producing three merged vertices
    \[
        M = (T_1,S_2), \quad B=(B_1,B_2), \quad C=(C_1,C_2).
    \]
    {Note that no new edges are introduced and no existing edges are removed.}
    Instead, every edge of $G_1$ or $G_2$ that was incident to one of the
    identified vertices is now incident to the corresponding merged vertex.
    The distinguished vertices of the resulting graph are $(S_1, T_2, B, C)$.
    The series operation is associative, and two examples are shown in
    Figure~\ref{fig:series}.
    
    	\begin{figure}[H]
    	\begin{center}
    		\resizebox{\linewidth}{!}{
    			\begin{tabular}{cc|ccc|cc}
    				\begin{tikzpicture}[scale=0.6]
    					\node[vertex] (S1) {$S_1$};
    					\node[vertex,fill=blue!20] (T1)  [below=of S1] {$T_1$};
    					\node[vertex,fill=red!20] (C1)  [right=of S1] {$C_1$};
    					\node[vertex,fill=cyan!20] (B1)  [below=of C1] {$B_1$};
    					\draw[arrow] (S1) -- (T1) node [midway, left] {};
    					\draw[arrow] (S1) -- (B1) node [midway, above] {};
    					\node[emptynode] (R1) [below=of B1] {};
    					
    					\node (z) [right=of B1] {\fontsize{52}{58}\sffamily\bfseries$\oseries{}{}$};
    					
    					\node[vertex,fill=blue!20] (S2) [right=of z] {$S_2$};
    					\node[vertex] (T2)  [below=of S2] {$T_2$};
    					\node[vertex,fill=red!20] (C2)  [right=of S2] {$C_2$};
    					\node[vertex,fill=cyan!20] (B2)  [below=of C2] {$B_2$};
    					\draw[arrow] (S2) -- (T2) ;
    					
    					\node (e) [right=of C2] {\fontsize{52}{58}\sffamily\bfseries$=$};
    					
    					\node[vertex,fill=blue!20,dashed] (M) [right=of e]{$M$};
    					\node[vertex] (S) [above=of M] {$S_1$};
    					\node[vertex] (T)  [below=of M] {$T_2$};
    					\node[vertex,fill=red!20,dashed] (C)  [right=of S] {$C$};
    					\node[vertex,fill=cyan!20,dashed] (B)  [below=of C] {$B$};
    					\draw[arrow] (S) -- (M) ;
    					\draw[arrow] (S) -- (B) ;
    					\draw[arrow] (M) -- (T) ;
    				\end{tikzpicture} \\ \\ \hline \hline \\
    				\begin{tikzpicture}[scale=0.6]
    					\node[vertex] (S1) {$S_1$};
    					\node[vertex,fill=blue!20] (T1)  [below=of S1] {$T_1$};
    					\node[vertex,fill=red!20] (C1)  [right=of S1] {$C_1$};
    					\node[vertex,fill=cyan!20] (B1)  [below=of C1] {$B_1$};
    					\draw[arrow] (S1) -- (T1) ;
    					\node[emptynode] (R1) [below=of B1] {};
    					
    					\node (z) [right=of B1] {\fontsize{52}{58}\sffamily\bfseries$\oseries{}{}$};
    					
    					\node[vertex,fill=blue!20] (S2) [right=of z] {$S_2$};
    					\node[vertex] (T2)  [below=of S2] {$T_2$};
    					\node[vertex,fill=red!20] (C2)  [right=of S2] {$C_2$};
    					\node[vertex,fill=cyan!20] (B2)  [below=of C2] {$B_2$};
    					\draw[arrow] (S2) -- (T2) node [midway, left] {};
    					\draw[arrow] (S2) -- (C2) node [midway, above] {};
    					
    					\node (e) [right=of C2] {\fontsize{52}{58}\sffamily\bfseries$=$};
    					
    					\node[vertex,fill=blue!20,dashed] (M)  [right=of e] {$M$};
    					\node[vertex] (S) [above=of M]{$S_1$};
    					\node[vertex] (T)  [below=of M] {$T_2$};
    					\node[vertex,fill=red!20,dashed] (C)  [right=of M] {$C$};
    					\node[vertex,fill=cyan!20,dashed] (B)  [below=of C] {$B$};
    					\draw[arrow] (S) -- (M);
    					\draw[arrow] (M) -- (T);
    					\draw[arrow] (M) -- (C);
    				\end{tikzpicture}
    				
    			\end{tabular}
    		}
    	\end{center}
    	\caption{Two examples of the series operation $\oseries{}{}$. In each case, the dashed vertices are new vertices obtained by merging the vertices of the same color in the original graphs. For example, $T_1$ and $S_2$ merge to form $M.$}
    	\label{fig:series}
    \end{figure}
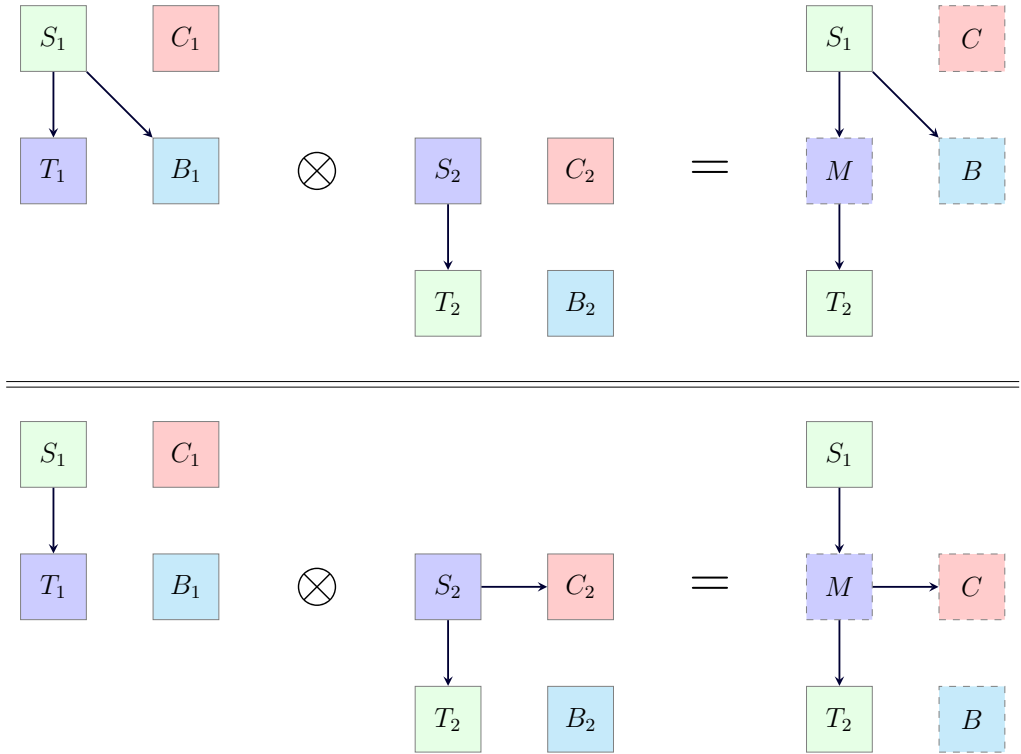

    \item \emph{Parallel Operation.}
    The graph $\oparal{G_1}{G_2}$ is obtained by taking the disjoint union of
    $G_1$ and $G_2$ and then {merging all four distinguished pairs of vertices}:
    \[
        S_1 \equiv S_2,\qquad T_1 \equiv T_2,\qquad
        B_1 \equiv B_2,\qquad C_1 \equiv C_2.
    \]
    These identifications yield the new vertices
    \[
        S=(S_1,S_2),\quad T=(T_1,T_2),\quad B=(B_1,B_2),\quad C=(C_1,C_2).
    \]
    As in the series case, {no edges are added or deleted}.
    All edges incident to any of the merged vertices are redirected to the new
    merged vertex, and edges from $G_1$ and $G_2$ that become duplicates are
    naturally merged in the union.
    The special vertex tuple of $\oparal{G_1}{G_2}$ is $(S,T,B,C)$.
    An example appears in Figure~\ref{fig:parallel}.

    \begin{figure}[H]
    	\begin{center}
    		\resizebox{\linewidth}{!}{
    			\begin{tikzpicture}[scale=0.6]
    				\node[vertex] (M1) {$M_1$};
    				\node[vertex,fill=blue!20] (S1) [above=of M1] {$S_1$};
    				\node[vertex,fill=yellow!20] (T1)  [below=of M1] {$T_1$};
    				\node[vertex,fill=cyan!20] (B1)  [right=of S1] {$B_1$};
    				\node[vertex,fill=red!20] (C1)  [below=of B1] {$C_1$};
    				\draw[arrow] (S1) -- (M1) ;
    				\draw[arrow] (S1) -- (B1) ;
    				\draw[arrow] (M1) -- (T1) ;
    				
    				\node (z) [right=of C1] {\fontsize{52}{58}\sffamily\bfseries$\oparal{}{}$};
    				
    				\node[vertex] (M2) [right=of z]  {$M_2$};
    				\node[vertex,fill=blue!20] (S2) [above=of M2]{$S_2$};
    				\node[vertex,fill=yellow!20] (T2)  [below=of M2] {$T_2$};
    				\node[vertex,fill=red!20] (C2)  [right=of M2] {$C_2$};
    				\node[vertex,fill=cyan!20] (B2)  [above=of C2] {$B_2$};
    				\draw[arrow] (S2) -- (M2);
    				\draw[arrow] (S2) -- (B2);
    				\draw[arrow] (M2) -- (T2);
    				\draw[arrow] (M2) -- (C2);
    				
    				\node (e) [right=of C2] {\fontsize{52}{58}\sffamily\bfseries$=$};
    				
    				\node[vertex] (aM1) [right=of e] {$M_1$};
    				\node[vertex,dashed,fill=blue!20] (aS) [above right=of aM1] {$S$};
    				\node[vertex] (aM2) [below right=of aS] {$M_2$};
    				\node[vertex,dashed,fill=yellow!20] (aT)  [below left=of aM2] {$T$};
    				\node[vertex,dashed,fill=red!20] (aC)  [right=of aM2] {$C$};
    				\node[vertex,dashed,fill=cyan!20] (aB)  [right=of aS] {$B$};
    				\draw[arrow] (aS) -- (aM1);
    				\draw[arrow] (aS) -- (aM2);
    				\draw[arrow] (aS) -- (aB);
    				\draw[arrow] (aM1) -- (aT);
    				\draw[arrow] (aM2) -- (aT);
    				\draw[arrow] (aM2) -- (aC);
    			\end{tikzpicture}
    		}
    	\end{center}
    	\caption{An example of the parallel operation $\oparal{}{}$. The dashed vertices are new vertices obtained by merging vertices of the same color in the original graphs. For example, $S_1$ and $S_2$ merge to form $S.$ Note that while the two original graphs were disjoint, the edge $(S, B)$ in the resulting graph is the merger of two edges $(S_1, B_1)$ and $(S_2, B_2)$ of the original graphs.}
    	\label{fig:parallel}
    \end{figure}
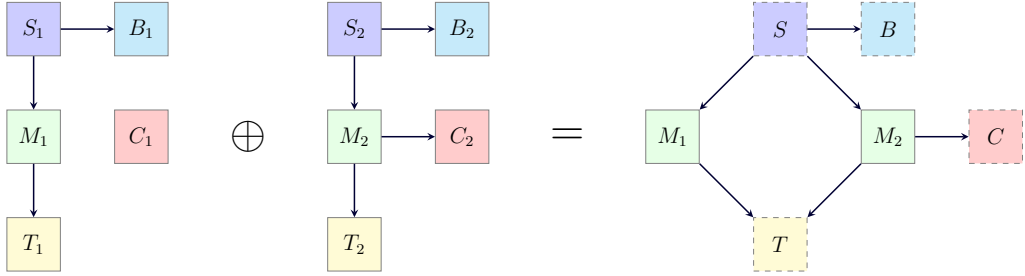

    \item \emph{Loop Operation.}
    The graph $\oloop{G_1}$ is formed by {adding four fresh vertices}
    \[
        S,\; T,\; B,\; C,
    \]
    which become the new distinguished start, terminate, break, and continue
    vertices, respectively.  We then {introduce five new edges} that connect
    these vertices to $G_1$:
    \[
        (S,S_1),\qquad (S,T),\qquad (T_1,S),\qquad (C_1,S),\qquad (B_1,T).
    \]
    No other edges of $G_1$ are modified.  Thus, the loop operation extends
    $G_1$ by creating a new “loop frame’’ around it while preserving all internal
    structure of $G_1$.  The special vertex tuple of $\oloop{G_1}$ is $(S,T,B,C)$,
    and an example is provided in Figure~\ref{fig:loop}.
    
    	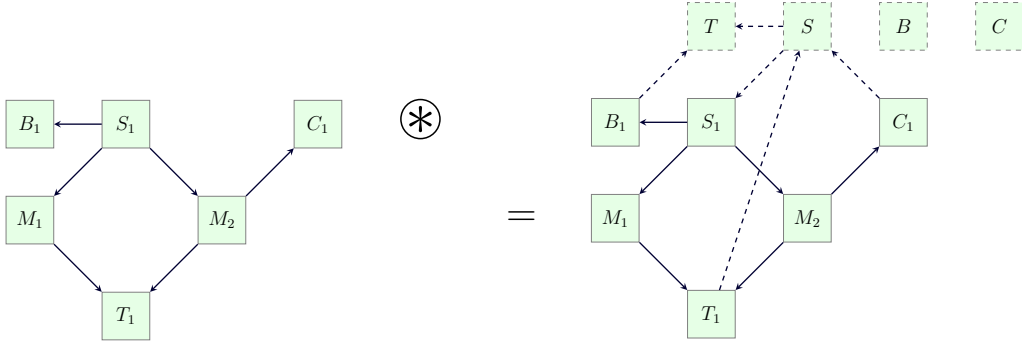
\begin{figure}[H]
    	\begin{center}
    		\resizebox{\linewidth}{!}{
    			\begin{tikzpicture}[scale=0.6]
    				\node[vertex] (aM1)  {$M_1$};
    				\node[vertex] (aS) [above right=of aM1] {$S_1$};
    				\node[vertex] (aM2) [below right=of aS] {$M_2$};
    				\node[vertex] (aT)  [below left=of aM2] {$T_1$};
    				\node[vertex] (aC)  [above right=of aM2] {$C_1$};
    				\node[vertex] (aB)  [left=of aS] {$B_1$};
    				\draw[arrow] (aS) -- (aM1);
    				\draw[arrow] (aS) -- (aM2);
    				\draw[arrow] (aS) -- (aB);
    				\draw[arrow] (aM1) -- (aT);
    				\draw[arrow] (aM2) -- (aT);
    				\draw[arrow] (aM2) -- (aC);
    				
    				\node (star) [right=of aC] {\fontsize{52}{58}\sffamily\bfseries$\oloop{}$};
    				
    				\node (e) [below right=of star] {\fontsize{52}{58}\sffamily\bfseries$=$};
    				
    				\node[vertex] (baM1) [right=of e] {$M_1$};
    				\node[vertex] (baS) [above right=of baM1] {$S_1$};
    				\node[vertex,dashed] (S) [above right=of baS] {$S$};
    				\node[vertex,dashed] (B) [right=of S] {$B$};
    				\node[vertex,dashed] (C) [right=of B] {$C$};
    				\node[vertex] (baM2) [below right=of baS] {$M_2$};
    				\node[vertex] (baT)  [below left=of baM2] {$T_1$};
    				\node[vertex] (baC)  [above right=of baM2] {$C_1$};
    				\node[vertex] (baB)  [left=of baS] {$B_1$};
    				\node[vertex,dashed] (T) [left=of S] {$T$};
    				\draw[arrow] (baS) -- (baM1);
    				\draw[arrow] (baS) -- (baM2);
    				\draw[arrow] (baS) -- (baB);
    				\draw[arrow] (baM1) -- (baT);
    				\draw[arrow] (baM2) -- (baT);
    				\draw[arrow] (baM2) -- (baC);
    				\draw[arrow,dashed] (S) -- (baS);
    				\draw[arrow,dashed] (S) -- (T);
    				\draw[arrow,dashed] (baT) -- (S);
    				\draw[arrow,dashed] (baB) -- (T);
    				\draw[arrow,dashed] (baC) -- (S);
    			\end{tikzpicture}
    		}
    	\end{center}
    	\caption{An example of the loop operation. Newly introduced vertices and edges are shown as dashed.}
    	\label{fig:loop}
    \end{figure}

\end{enumerate}

We say that an $\slap$ graph $G=(V, E, S, T, B, C)$ is closed if there are no incoming edges to either $B$ or $C.$

Given an $\slap$ graph $G=(V, E, S, T, B, C)$, its \emph{grammatical decomposition} 
is the parse tree of $G$ according to grammar~\eqref{gram:graph}. 
In other words, the grammatical decomposition is a tree whose leaves are 
atomic $\slap$ graphs, and whose internal nodes correspond to applications 
of the three SPL operations 
\(\oseries{c_1}{c_2}\), \(\oparal{c_1}{c_2}\), and \(\oloop{c_1}\) 
to their children \(c_1\) and \(c_2\).  

Intuitively, each node of the decomposition tree represents a well-defined
substructure of the program’s control-flow graph.  
A leaf encodes the smallest indivisible fragment—typically a single edge or
basic block boundary—while an internal node represents a larger region obtained
by composing its child regions either in sequence (series), in parallel
(branching/merging), or as a loop body.  
Thus, every subtree corresponds to a contiguous portion of the original program, 
and moving upward in the tree corresponds to assembling progressively larger
program fragments.  
When we reach the root of the grammatical decomposition, the entire control-flow
graph \(G\) has been reconstructed, meaning the root node corresponds exactly to
the whole program.

For example, Figure~\ref{fig:decompo} (bottom left) shows the grammatical 
decomposition of the $\slap$ graph in Figure~\ref{fig:decompo} (bottom right).  
Each subtree of this decomposition represents a meaningful fragment of the CFG,
and the root represents the complete program structure.

There is a natural surjective homomorphism $\cfg{\cdot}$ between programs and $\slap$ graphs, defined as follows:
$$
\begin{matrix}
\cfg{\epsilon} = A_\epsilon & \cfg{\texttt{break}} = A_{\texttt{break}} & \cfg{\texttt{continue}} = A_{\texttt{continue}}
\end{matrix}
$$
$$
	\cfg{P_1 ; P_2} = \oseries{\cfg{P_1}}{\cfg{P_2}}
$$
$$
	\cfg{\texttt{if} ~\varphi~ \texttt{then} ~P_1~ \texttt{else} ~P_2~ \texttt{fi}} = \oparal{\cfg{P_1}}{\cfg{P_2}}
$$
$$
\cfg{\texttt{while} ~\varphi~ \texttt{do} ~P_1~ \texttt{od}} = \oloop{\cfg{P_1}}
$$

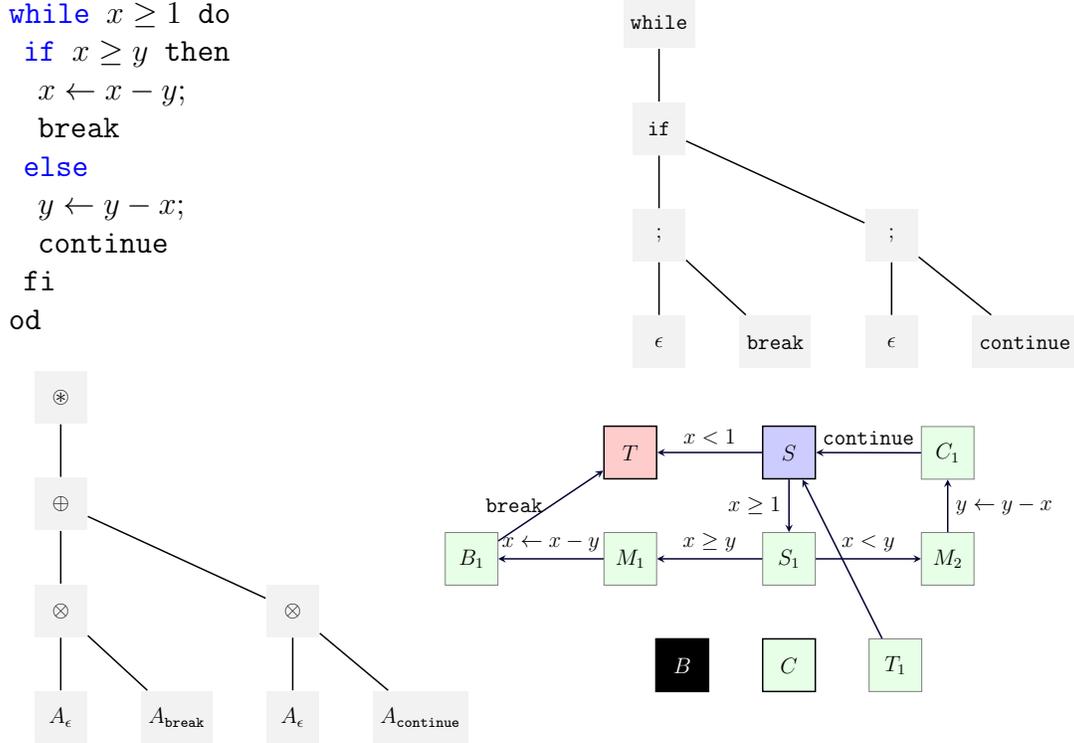
\begin{figure}
	\hspace{-10pt}\begin{subfigure}[]{0.3\textwidth}
		\begin{lstlisting}[mathescape,numbers=none]
while $x \geq 1$ do
	if $x \geq y$ $\text{then}$
		$x \gets x - y;$
		$\text{break}$
	else
		$y \gets y - x;$
		$\text{continue}$
	fi
od
		\end{lstlisting} 
	\end{subfigure}
	\hfill
	\begin{subfigure}[]{0.2\textwidth}
		\scalebox{0.7}{
		\begin{tikzpicture}[scale=0.4]
			\node[rect] (eps1) {$\epsilon$};
			\node[rect] (break) [right=of eps1] {$\texttt{break}$};
			\node[rect] (eps2) [right=of break] {$\epsilon$};
			\node[rect] (continue) [right=of eps2] {$\texttt{continue}$};
			\node[rect] (sem1) [above=of eps1] {$;$};	
			\node[rect] (sem2) [above=of eps2] {$;$};	
			\node[rect] (if) [above=of sem1] {$\texttt{if}$};	
			\node[rect] (while) [above=of if] {$\texttt{while}$};	
			\draw[thick] (while) -- (if) ;
			\draw[thick] (if) -- (sem1) ;
			\draw[thick] (if) -- (sem2) ;
			\draw[thick] (sem1) -- (eps1) ;
			\draw[thick] (sem1) -- (break) ;
			\draw[thick] (sem2) -- (eps2) ;
			\draw[thick] (sem2) -- (continue) ;
		\end{tikzpicture}
	}
	\end{subfigure}
	\hspace{3cm}
	\\
	\hspace{-10pt}\begin{subfigure}[]{0.2\textwidth}
		\scalebox{0.7}{
		\begin{tikzpicture}[scale=0.4]
			\node[rect] (eps1) {$A_\epsilon$};
			\node[rect] (break) [right=of eps1] {$A_\texttt{break}$};
			\node[rect] (eps2) [right=of break] {$A_\epsilon$};
			\node[rect] (continue) [right=of eps2] {$A_\texttt{continue}$};
			\node[rect] (sem1) [above=of eps1] {$\oseries{}{}$};	
			\node[rect] (sem2) [above=of eps2] {$\oseries{}{}$};	
			\node[rect] (if) [above=of sem1] {$\oparal{}{}$};	
			\node[rect] (while) [above=of if] {$\sloop$};	
			\draw[thick] (while) -- (if) ;
			\draw[thick] (if) -- (sem1) ;
			\draw[thick] (if) -- (sem2) ;
			\draw[thick] (sem1) -- (eps1) ;
			\draw[thick] (sem1) -- (break) ;
			\draw[thick] (sem2) -- (eps2) ;
			\draw[thick] (sem2) -- (continue) ;
		\end{tikzpicture}}
	\end{subfigure}
	\hfill
	\begin{subfigure}[]{0.3\textwidth}
		\scalebox{0.7}{
		\begin{tikzpicture}[scale=0.4,node distance=1cm and 2cm]
			
			\node[vertex] (baM1)  {$M_1$};
			\node[vertex] (baB)  [left=of baM1] {$B_1$};
			\node[vertex] (baS) [right=of baM1] {$S_1$};
			\node[vertex] (baM2) [right=of baS] {$M_2$};
			\node[vertext] (T) [above=of baM1] {$T$};
			
			\node[vertexs] (S) [above=of baS] {$S$};
			\node[vertex] (baC)  [above=of baM2] {$C_1$};

			\node[vertexc] (C) [below=of baS] {$C$};
			\node[vertexb] (B) [left=1cm of C] {$B$};
			
			\node[vertex] (baT)  [right=1cm of C] {$T_1$};

			\draw[arrow] (S) -- (T) node [midway, above] {$x<1$};
			\draw[arrow] (S) -- (baS) node [midway, left] {$x\geq 1$};
			\draw[arrow] (baS) -- (baM1) node [midway, above] {$x \geq y$};
			\draw[arrow] (baS) -- (baM2) node [midway, above] {$x < y$};
			\draw[arrow] (baM1) -- (baB) node [midway, above] {$x \gets x - y$};
			\draw[arrow] (baM2) -- (baC) node [midway, right] {$y \gets y - x$};
			\draw[arrow] (baB) -- (T) node [midway, left] {$\texttt{break}$};
			\draw[arrow] (baC) -- (S) node [midway, above] {$\texttt{continue}$};
			\draw[arrow] (baT) -- (S);
		\end{tikzpicture}}
	\end{subfigure}
	\hspace{4cm}
	\caption{A program $P$ (top left), its parse tree (top right), the corresponding parse tree of $G = \cfg{P},$ aka the grammatical decomposition of $G$ (bottom left) and the graph $G = \cfg{P}$ (bottom right). The edges of the graph are labeled according to the commands/conditions of the program.}
	\label{fig:decompo}
\end{figure}

It is easy to verify that this homomorphism matches the usual definition of control-flow graphs of programs and that $P$ is closed if and only if $\cfg{P}$ is closed. Thus, the set of $\slap$ graphs is precisely the same as the set of control-flow graphs of structured programs.

Moreover, given a program $P,$ we can simply parse it according to the grammar in~\eqref{gram:prog} and obtain a parse tree, then apply our homomorphism to the parse tree to find an equivalent parse tree of its control-flow graph based on the grammar in~\eqref{gram:graph}. This is exactly the grammatical decomposition of the CFG, which shows us how the control-flow graph of the current program can be obtained by merging the control-flow graphs of its smaller parts using one of our three operations. Since parsing a context-free grammar takes linear time in the size of the program, i.e.~$O(n),$ thus we obtain a grammatical decomposition of our control-flow graph in $O(n)$ time, too. Figure~\ref{fig:decompo} shows an example. Note that in this example, we are also labeling the edges of our control-flow graph using commands/conditions in the program. The labels are not part of the graph. They are included to improve readability.

In our experiments in Section~\ref{sec:exp}, we find the grammatical decompositions by parsing the program. However, it is possible to find them directly based on the CFG, too. Formally, the problem of deciding whether a graph $G$ is $\slap$ and finding its grammatical decomposition is equivalent to parsing a context-free grammar. Thus, the CYK algorithm is applicable~\cite{sakai1961syntax}.

We note that our grammatical decompositions are conceptually more similar to series-parallel graphs than tree decompositions and treewidth-based algorithms. Specifically, we do not have a separate parameter such as treewidth or pathwidth. Instead, we only focus on the set of graphs that can be obtained using the three operations mentioned above. This is how we sidestep the NP-hardness of register allocation in the next section. Instead of solving it over all graphs, we solve it only on $\slap$ graphs, which are the same as CFGs and thus cover all real-world instances of the problem. Intuitively, the main reason for the efficiency of treewidth/pathwidth-based algorithms is that they are able to repeatedly find small cuts in the CFG and use them to recursively cut it into small parts. This is also the insight behind our approach. The difference is that treewidth-based approaches find cuts of size at most $8,$ and pathwidth-based approaches find cuts of size $2 \cdot d + 1$ where $d$ is the nesting depth, whereas our cuts are precisely our special vertices and thus of size $4.$ As we shall soon see, smaller cuts lead to more efficient algorithms.

\section{Register Allocation}

As an application of our concept of grammatical decomposition, we consider the classical problem of minimum-cost register allocation as formalized in~\cite{DBLP:conf/cc/Krause13}. A cost is assigned to each allocation of variables to registers, which is supposed to model the time wasted on spills or rematerialization. We note that this is a more general formulation of the problem than those of~\cite{THORUP1998159,DBLP:journals/pacmpl/ConradoGL23} which only focus on deciding whether it is possible to avoid spilling altogether and obtain a cost of zero. See Section~\ref{sec:last}. In~\cite{DBLP:conf/cc/Krause13}, a treewidth-based algorithm is provided which obtains an optimal register allocation in XP time $O(|G| \cdot |\vars|^{2 \cdot (t + 1) \cdot r}),$ where $G$ is the control-flow graph, $\vars$ is the set of program variables, $t$ is the treewidth of $G$ and $r$ is the number of available registers. If both $r$ and $t$ are constants, then this is a polynomial-time algorithm. However, since the treewidth of programs in languages such as C can be up to 7~\cite{DBLP:journals/dam/KrauseLS20}, this algorithm's worst-case runtime over CFGs is $O(|G| \cdot |\vars|^{\textcolor{red}{16 \cdot r}}).$ In this section, we present an alternative algorithm using our grammatical decomposition which provides a significant runtime improvement and runs in time $O(|G| \cdot |\vars|^{\textcolor{red}{5 \cdot r}}).$ This is a huge asymptotic improvement over the algorithm of~\cite{DBLP:journals/dam/KrauseLS20}. 

\subsection{Problem Definition}

Suppose we are given a program $P$ with control-flow graph $G=\cfg{P} = (V, E, S, T, B, C).$ Let $[r] = \{0, 1, \ldots, r-1\}$ be the set of available registers and $\vars$ the set of our program variables. Every variable $v \in \vars$ has a lifetime $\lt{v}$ which is a connected subgraph of $G.$ See~\cite{poletto1999linear} for a more detailed treatment of lifetimes. Since lifetimes can be computed by a simple data-flow analysis, we assume without loss of generality that they are given as inputs to our algorithm. For a vertex $v$ or edge $e$ of $G,$ we denote the set of variables that are alive at this vertex/edge by $L(v)$ or $L(e).$ An \emph{assignment} is a function $f: \vars \rightarrow [r] \cup \{\perp\}$ which maps each variable either to a register or to $\perp.$ The latter models the variable being spilled. An assignment is valid if it does not map two variables with intersecting lifetimes to the same register. We denote the set of all valid assignments by $F.$

The interference graph of our program $P$ is a graph $\mathbb{I} = (\vars, E_{\mathbb I})$ in which there is one vertex for each program variable and there is an edge $\{u, v\}$ if the variables $u$ and $v$ can be alive at the same time, i.e.~$\lt{u} \cap \lt{v} \neq \emptyset.$ Any valid assignment $f$ is a valid coloring of a subset of vertices of $\mathbb{I}$ with colors in $[r].$ This correspondence between register allocation and graph coloring is well-known and due to Chaitin~\cite{DBLP:conf/pldi/Chaitin82}. We note that for every vertex $v,$ the set $L(v)$ of variables alive at $v$ forms a clique in $\mathbb I.$

We provide an example taken from~\cite{example}. Figure~\ref{fig:regaloc} shows a program $P$ and its control-flow graph $G = \cfg{P},$ including live variables at each vertex, and the interference graph $\mathbb{I}.$ Our goal is to color a subset of vertices of $\mathbb{I}$ with $r$ colors, where $r$ is the number of available registers. A complete coloring with $4$ colors is shown in the figure. This avoids any spilling. We also show a partial coloring with $3$ colors and some spilling.

\begin{figure}
	\begin{subfigure}{0.33\textwidth}
		\begin{lstlisting}[mathescape,numbers=none]
while $\varphi_1$ do
	$a \gets b+c$;
	$d \gets -a$;
	$e \gets d+f$;
	if $\varphi_2$ $\texttt{then}$
		$f \gets 2 \cdot e$
	else
		$b \gets d + e$;
		$e \gets e - 1$
	fi
	$b \gets f + c$
od
		\end{lstlisting} 
	\end{subfigure}
	\begin{subfigure}{0.33\linewidth}
		\begin{center}
			\scalebox{0.8}{
		\begin{tikzpicture}[scale=0.6]
			\node[vertex] (s)  {$S$};
			\node[vertexla] (ls) [above=0cm of s]  {$\{b, c, f\}$};
			\node[vertex] (t) [right=of s] {$T$};
			\node[vertexla] (lt) [above=0cm of t]  {$\emptyset$};
			\node[vertex] (b) [right=of t]  {$B$};
			\node[vertexla] (lb) [above=0cm of b]  {$\emptyset$};
			\node[vertex] (c) [right=of b]  {$C$};
			\node[vertexla] (lc) [above=0cm of c]  {$\emptyset$};
			\node[vertex] (c1) [below=of t]  {$C_1$};
			\node[vertexla] (lc1) [below=0cm of c1]  {$\{b, c, f\}$};
			\node[vertex] (b1) [right=of c1]  {$B_1$};
			\node[vertexla] (lb1) [below=0cm of b1]  {$\emptyset$};
			\node[vertex] (1) [below=of s]  { };
			\node[vertexla] (l1) [left=0cm of 1]  {$\{b, c, f\}$};
			\node[vertex] (2) [below=of 1]  { };
			\node[vertexla] (l2) [left=0cm of 2]  {$\{a, c, f\}$};
			\node[vertex] (3) [below=of 2]  { };
			\node[vertexla] (l3) [left=0cm of 3]  {$\{c, d, f\}$};
			\node[vertex] (if) [below=of 3]  { };
			\node[vertexla] (lif) [left=0cm of if]  {$\{c, d, e, f\}$};
			\node[vertex] (then) [right=3cm of if]  { };
			\node[vertexla] (lthen) [right=0cm of then]  {$\{c, e, f\}$};
			\node[vertex] (fi) [below=of if] {};
			\node[vertexla] (lfi) [left=0cm of fi]  {$\{c, f\}$};
			\node[vertex] (5) [below=of fi] {};
			\node[vertexla] (l5) [left=0cm of 5]  {$\{b, c, f\}$};
			
			\draw[arrow] (1) -- (2) node [midway, left] {$a \gets b+c$};
			\draw[arrow] (2) -- (3) node [midway, left] {$d \gets -a$};
			\draw[arrow] (3) -- (if) node [midway, left] {$e \gets d+f$};
			\draw[arrow] (if) -- (then) node [midway, above] {$\neg\varphi_2, b \gets d+e$};
			\draw[arrow] (if) -- (fi) node [midway, left] {$\varphi_2, f \gets 2 \cdot e$};
			\draw[arrow] (then) -- (fi) node [midway, below right] {$e \gets e-1$};
			\draw[arrow] (fi) -- (5) node [midway, right] {$b \gets f+c$};
			\draw[arrow] (s) -- (1) node [midway, left] {$\varphi_1$};
			\draw[arrow] (s) -- (t) node [midway, above] {$\neg \varphi_1$};
			\draw [arrow] (5) to [out=150,in=210] (s);
			%\draw[arrow, bend left=90] (5) -- (s);
			\draw[arrow] (b1) -- (t);
			\draw[arrow] (c1) -- (s);
		\end{tikzpicture}}
	\end{center}
	\end{subfigure}

\vspace{1em}
\begin{subfigure}{0.3\textwidth}
	\scalebox{0.72}{
	\begin{tikzpicture}[scale=0.6]
		\node[vertex] (a)  {$a$};
		\node[vertex] (b) [below=of a] {$b$};
		\node[vertex] (c) [right=of a] {$c$};
		\node[vertex] (f) [right=of b] {$f$};
		\node[vertex] (e) [right=of c] {$e$};
		\node[vertex] (d) [right=of f] {$d$};
		\draw[thick] (b) -- (c);
		\draw[thick] (b) -- (f);
		\draw[thick] (f) -- (c);
		\draw[thick] (a) -- (c);
		\draw[thick] (a) -- (f);
		\draw[thick] (d) -- (c);
		\draw[thick] (d) -- (f);
		\draw[thick] (e) -- (c);
		\draw[thick] (e) -- (d);
		\draw[thick] (e) -- (f);
	\end{tikzpicture}}
\end{subfigure}
\hfill
\begin{subfigure}{0.3\textwidth}
	\scalebox{0.72}{
	\begin{tikzpicture}[scale=0.6]
		\node[vertex, fill=blue!40] (a)  {$a$};
		\node[vertex, fill=green!40] (b) [below=of a] {$b$};
		\node[vertex, fill=red!40] (c) [right=of a] {$c$};
		\node[vertex, fill=yellow!40] (f) [right=of b] {$f$};
		\node[vertex, fill=green!40] (e) [right=of c] {$e$};
		\node[vertex, fill=blue!40] (d) [right=of f] {$d$};
		\draw[thick] (b) -- (c);
		\draw[thick] (b) -- (f);
		\draw[thick] (f) -- (c);
		\draw[thick] (a) -- (c);
		\draw[thick] (a) -- (f);
		\draw[thick] (d) -- (c);
		\draw[thick] (d) -- (f);
		\draw[thick] (e) -- (c);
		\draw[thick] (e) -- (d);
		\draw[thick] (e) -- (f);
	\end{tikzpicture}}
\end{subfigure}
\hfill
\begin{subfigure}{0.3\textwidth}
	\scalebox{0.72}{
	\begin{tikzpicture}[scale=0.6]
		\node[vertex] (a)  {$a$};
		\node[vertex, fill=green!40] (b) [below=of a] {$b$};
		\node[vertex, fill=red!40] (c) [right=of a] {$c$};
		\node[vertex] (f) [right=of b] {$f$};
		\node[vertex, fill=green!40] (e) [right=of c] {$e$};
		\node[vertex, fill=blue!40] (d) [right=of f] {$d$};
		\draw[thick] (b) -- (c);
		\draw[thick] (b) -- (f);
		\draw[thick] (f) -- (c);
		\draw[thick] (a) -- (c);
		\draw[thick] (a) -- (f);
		\draw[thick] (d) -- (c);
		\draw[thick] (d) -- (f);
		\draw[thick] (e) -- (c);
		\draw[thick] (e) -- (d);
		\draw[thick] (e) -- (f);
	\end{tikzpicture}}
\end{subfigure}
\caption{An Example Program $P$ (top left), its control-flow graph $G = \cfg{P}$ (top right) in which every vertex is labeled by its set of live variables in red, the interference graph $\mathbb I$ (bottom left), a coloring of all vertices of $\mathbb{I}$ with 4 colors corresponding to allocating all variables to 4 registers (bottom center), and a coloring of a subset of vertices of $\mathbb{I}$ with 3 colors corresponding to spilling the variables $a$ and $f$ (bottom right).}
\label{fig:regaloc}
\end{figure}
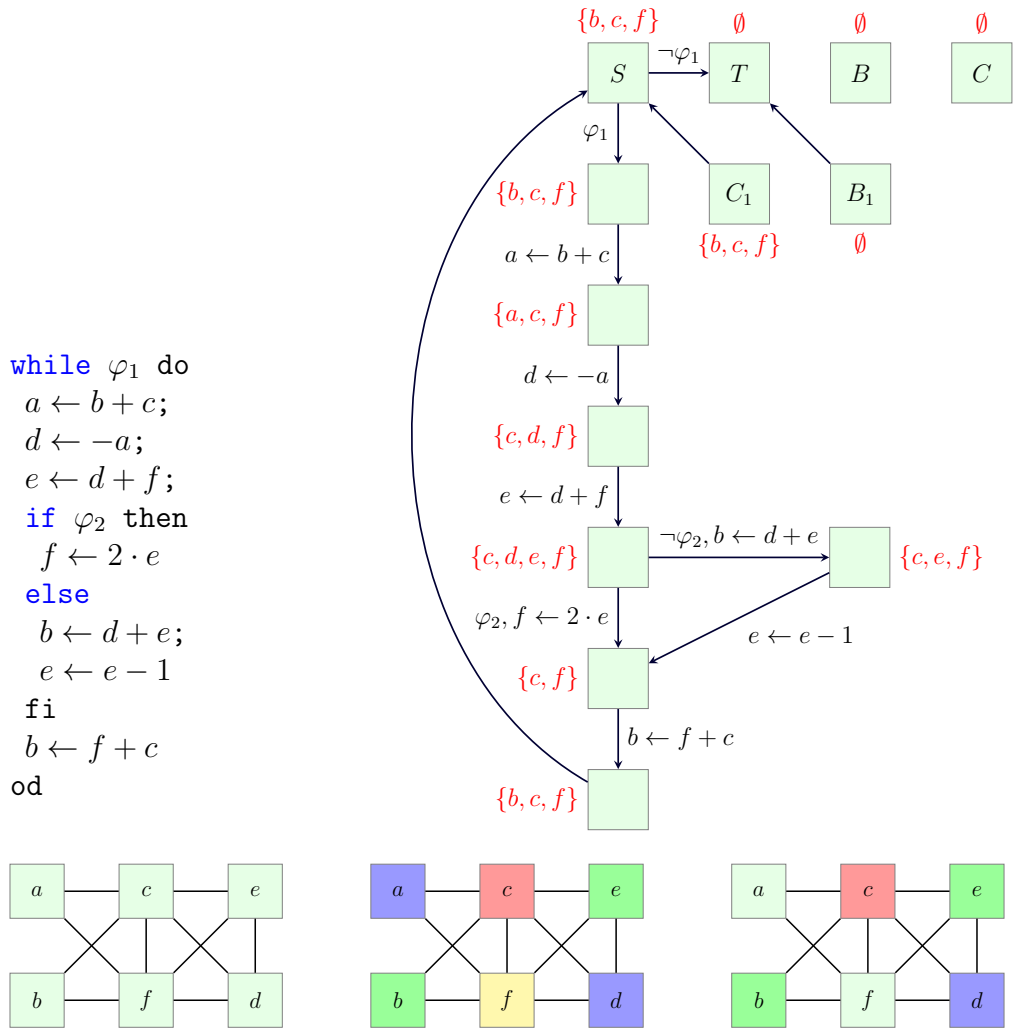

A \emph{cost function}~\cite{DBLP:conf/cc/Krause13} is a function $c: E \times F \rightarrow [0, \infty).$ For an edge $e \in E$ of the control-flow graph, which corresponds to one command of the program, $c(e, f)$ is the cost of running this command when the registers are allocated as per $f.$ We assume that $c(e, f)$ only depends on the allocation decisions for variables that are alive at $e.$ Following~\cite{DBLP:conf/cc/Krause13}, we further assume constant-time oracle access to evaluations of $c.$ In practice, $c$ is often obtained by profiling. Different optimization goals, such as total runtime or code size, may be modeled by choosing a suitable function $c.$

The problem of optimal register allocation provides $P, G, r, \vars, c$ and the lifetimes of variables as input and asks for an assignment $f \in F$ that minimizes the total cost, i.e.~our goal is to $$\text{minimize } \sum_{e \in E} c(e, f)$$ by choosing the best possible $f.$

\subsection{Our Algorithm} \label{sec:algo}

We now show how to perform dynamic programming on the grammatical decomposition of our control-flow graph $G$ to find an optimal register allocation. Our algorithm is quite simple and elegant. We process our grammatical decomposition in a bottom-up fashion and for every subgraph $H = (V_H, E_H, S_H, T_H, B_H, C_H)$ appearing in the grammatical decomposition, define the following dynamic programming variables:
$$
\begin{array}{rl}
\opt{H, f'} = & \text{The minimum total cost $\sum_{e \in E_H} c(e, f)$}\\
& \text{of an assignment $f$ over $H$ such that} \\ 
& f_{|L(S_H) \cup L(T_H) \cup L(B_H) \cup L(C_H)}=f'.
\end{array}$$
Intuitively, for every possible assignment $f'$ of the variables that are alive at any of the distinguished vertices $(S_H, T_H, B_H, C_H),$ we are asking for the minimum total cost of an assignment $f$ over all variables that agrees with $f'$ and extends it. After we compute our $\opt{\cdot, \cdot}$ values, the final answer of the algorithm, i.e.~the minimum cost of a register allocation, is simply $\min_{f} \opt{G, f}.$

We now show how to compositionally compute $\opt{H, f'}$ for any $\slap$ graph $H$ assuming that we have already computed $\opt{\cdot, \cdot}$ values for the $\slap$ subgraphs of $H.$ This is done by casework:
\begin{itemize}
	\item \emph{Atomic Graphs:} 	
	If $H \in \{A_\epsilon, A_{\texttt{break}}, A_{\texttt{continue}}\},$ then $H$ does not have any vertices other than the distinguished vertices $(S_H, T_H, B_H, C_H).$ Thus, all variables that are alive at any point in $H$ are also alive at one of the distinguished vertices and we simply set $\opt{H, f'} = c(e, f')$ for every partial allocation $f'.$ Here, $e$ is the unique edge in $H.$
\end{itemize}
	
	\para{Compatible Assignments} We say two partial assignments $f_1: \vars_1 \rightarrow [r] \cup \{\perp\}$ and $f_2: \vars_2 \rightarrow [r] \cup \{\perp\}$ are compatible and write $f_1 \leftrightarrows f_2$ if $\forall v \in \vars_1 \cap \vars_2 \quad f_1(v) = f_2(v).$
	Informally, $f_1$ and $f_2$ never make conflicting decisions on any variable $v$ but we have no restrictions on variables that are decided only by $f_1$ or only by $f_2.$ In other words, $f_1$ and $f_2$ can be combined in the same total assignment.

\begin{itemize}
	\item \emph{Series Operation:} If $H = \oseries{H_1}{H_2},$ then we have
	
	 $$\opt{\oseries{H_1}{H_2}, f'} =$$ $$\min_{\begin{matrix}f' \leftrightarrows f_1'\\ f' \leftrightarrows f_2'\\ f_1' \leftrightarrows f_2'\end{matrix}} \left(\opt{H_1, f_1'} + \opt{H_2, f_2'} - \sum_{e \in E_{H_1} \cap E_{H_2}} c(e, f_1') \right).$$

	The correctness of this calculation is an immediate corollary of the definition of our series operation. By construction, we have $L(B_{H_1}) = L(B_{H_2}) = L(B_H)$ and $L(C_{H_1}) = L(C_{H_2}) = L(C_H)$ and also $L(T_{H_1}) = L(S_{H_2}).$ Moreover, every edge of $H_1$ and $H_2$ is preserved in $\oseries{H_1}{H_2}.$ Thus, the total cost is simply the sum of costs in the two components. We should also be careful not to double-count the cost of edges that appear in both $H_1$ and $H_2$. Thus, we subtract these. Of course, the partial assignments $f', f_1'$ and $f_2'$ should be pairwise compatible.
	
	\item \emph{Parallel Operation:} This case is handled exactly as in the series case:
	$$
	\displaystyle \opt{\oparal{H_1}{H_2}, f'} =$$ $$\displaystyle \min_{\begin{matrix}f' \leftrightarrows f_1'\\ f' \leftrightarrows f_2'\\ f_1' \leftrightarrows f_2'\end{matrix}} \left(\opt{H_1, f_1'} + \opt{H_2, f_2'} - \sum_{e \in E_{H_1} \cap E_{H_2}} c(e, f_1') \right).
	$$
	This is because our parallel operation also preserves all the edges in $H_1$ and $H_2.$ Note that we might have edges that appear in both $H_1$ and $H_2,$ e.g.~we might have both $(S_{H_1}, B_{H_1})$ and $(S_{H_2}, B_{H_2})$ which are the same as the edge $(S_H, B_H).$ Thus, the total cost is the sum of costs in the components $H_1$ and $H_2$ minus the cost of their common edges. As before, we should also ensure that the partial assignments are all compatible.
	
	\item \emph{Loop Operation:} Suppose $H = \oloop{H_1}.$ In this case, by our construction, $H$ has the same vertices and edges as $H_1$ except for the introduction of the four new distinguished vertices $(S_H, T_H, B_H, C_H)$ and five new edges $e_1 = (S_H, S_{H_1}), e_2 = (S_H, T_H), e_3 = (T_{H_1}, S_H), e_4 = (C_{H_1}, S_H)$ and $e_5 = (B_{H_1}, T_H).$ Thus, our total cost is simply the total cost in $H_1$ plus the cost incurred at these new edges. Therefore, we have:
	$$
	\opt{\oloop{H_1}, f'} = \min_{f_1' \leftrightarrows f'} \left( \opt{H_1, f_1'} + \sum_{i=1}^5 c(e_i, f' \cup f_1') \right).
	$$
\end{itemize}
This concludes our algorithm which computes the cost of an optimal assignment $f.$ As is standard in dynamic programming approaches, $f$ itself can be obtained by retracing the steps of the algorithm and remembering the choices that led to the minimum values at every step.

\begin{theorem} \label{thm:main}
	Given a program $P$ with variables $\vars$ and control-flow graph $G=\cfg{P},$ the number $r$ of available registers and a cost function $c(\cdot, \cdot)$ as input, our algorithm above finds an optimal allocation of registers, i.e.~an optimal assignment function $f,$ in time $O(|G| \cdot |\vars|^{5 \cdot r}).$
\end{theorem}
\begin{proof}
	Correctness was argued above. We do a casework for runtime analysis:
	\begin{itemize}
		\item At atomic graphs, we are considering partial assignments $f'$ over variables that are alive at any of the four distinguished vertices. Let $a$ be one of these distinguished vertices. The set $L(a)$ of alive variables at $a$ forms a clique in the interference graph $\mathbb I.$ Thus, any valid $f'$ can assign $f'(v) \neq \perp$ to at most $r$ variables $v$ in $L(a).$ Moreover, no two variables can be assigned to the same register. Given that $|L(a)| \leq |\vars|,$ the total number of possible assignments for variables in $L(a)$ is at most $$\binom{|\vars|}{r} \cdot r! + \binom{|\vars|}{r-1} \cdot (r-1)! + \dots + \binom{|\vars|}{0} \cdot 0! \in O(r \cdot |\vars|^r).$$ Thus, the total number of $f'$ functions is at most $O(r^4 \cdot |\vars|^{4 \cdot r})$ given that we have four distinguished vertices. Our algorithm spends a constant amount of time for each $f',$ simply querying the cost of a single edge.
		\item When $H = \oseries{H_1}{H_2},$ we note that we have $B_H = B_{H_1} = B_{H_2}$ and $C_H = C_{H_1} = C_{H_2}.$ Similarly, we have $T_{H_1} = S_{H_2}.$ Thus, $f', f_1'$ and $f_2'$ need to jointly choose a register assignment for the variables that are alive at one of five vertices: $S_{H_1}, T_{H_1}, T_{H_2}, B$ and $C.$ An argument similar to the previous case shows that there are $O(r^5 \cdot |\vars|^{5 \cdot r})$ such assignments. We also note that $E_{H_1} \cap E_{H_2}$ has $O(1)$ many edges since any such edge must be connecting two distinguished vertices and we have only four such vertices. Thus, the total runtime here is also $O(r^5 \cdot |\vars|^{5 \cdot r}).$
		\item  When $H = \oparal{H_1}{H_2},$ a similar argument applies. In this case, we have $S_H = S_{H_1} = S_{H_2}, T_H = T_{H_1} = T_{H_2}, B_H = B_{H_1} = B_{H_2}$ and $C_H = C_{H_1} = C_{H_2}.$ Thus, we need to look at assignments for live variables at only four different vertices and our runtime is $O(r^4 \cdot |\vars|^{4 \cdot r}).$
		
		\item Finally, when $H = \oloop{H_1},$ we are introducing four new distinguished vertices. So, it seems that we have to consider the live variables at eight vertices in total, i.e.~the distinguished vertices of both $H$ and $H_1.$ However, note that $B_{H_1}$ has only one outgoing edge in our control-flow graph $G$ which goes to $T_H.$ Thus, we have $L(B_{H_1}) \subseteq L(T_H).$ For similar reasons, $L(T_{H_1}) \subseteq L(S_H)$ and $L(C_{H_1}) \subseteq L(S_H).$ Therefore, we only need to consider the program variables that are alive at one of the five vertices $S_H, T_H, B_H, C_H$ and $S_{H_1}.$ An argument similar to the previous cases shows that our runtime is $O(r^5 \cdot |\vars|^{5 \cdot r}).$
	\end{itemize}
	Finally, our algorithm has to process the grammatical decomposition in a bottom-up manner and compute the $\opt{\cdot, \cdot}$ values at every node. We have $O(|G|)$ nodes. Thus, the total worst-case runtime is $O(|G| \cdot r^5 \cdot |\vars|^{5 \cdot r}).$ Following~\cite{DBLP:conf/cc/Krause13} and other works on minimum-cost register allocation, we assume that $r$ is a constant. Thus, our runtime is $O(|G| \cdot |\vars|^{5 \cdot r}).$
\end{proof}

\para{Parallelization} We note that our algorithm is perfectly parallelizable since at every $\slap$ subgraph $H,$ one can compute the $\opt{H, f}$ values for different $f$ functions in parallel. Thus, if we have $k$ threads and $k$ is less than the number of possible partial functions $f,$ then our runtime is reduced to $O\left(\frac{|G| \cdot |\vars|^{5 \cdot r}}{k}\right).$

\subsection{Spill-free Register Allocation} \label{sec:last}

We remark that the works \cite{DBLP:conf/soda/BodlaenderGT98} (SODA 1998) and \cite{DBLP:journals/pacmpl/ConradoGL23} (OOPSLA 2023) provide linear-time algorithms with respect to the size of the CFG for the decision problem of existence of a spill-free register allocation, i.e.~setting $c(e, f) = 0$ if $f$ is valid and allocates all variables to registers meaning that it does not map anything to $\perp$ and $c(e, f) = +\infty$ otherwise, and simply asking whether an assignment with zero total cost is attainable. This is a special case of the problem we considered above. Unlike the general case, works on this special case often do not consider $r$ to be a constant and analyze their runtimes based on both $|G|$ and $r.$ The former work provides an algorithm with a runtime of $O(|G| \cdot r^{2 \cdot t \cdot r + 2\cdot r})$ where $t$ is the treewidth of the control-flow graph. The latter uses a different parameter, namely pathwidth, and obtains a runtime of $O(|G| \cdot p \cdot r^{p \cdot r+r+1}).$
Given that the treewidth of a structured program in languages such as C can be up to 7~\cite{DBLP:journals/dam/KrauseLS20} and the pathwidth is also empirically observed to be no more than 17 in~\cite{DBLP:journals/pacmpl/ConradoGL23}, these approaches provide runtimes of $O(|G| \cdot r^{\textcolor{red}{16 \cdot r}})$ and $O(|G| \cdot r^{\textcolor{red}{18 \cdot r + 1}})$ respectively. Moreover,~\cite{DBLP:journals/pacmpl/ConradoGL23} (Figure~9) observes that the vast majority of real-world programs have a pathwidth of 6 or lower. For these instances, their runtime would be $O(|G| \cdot r^{\textcolor{red}{7\cdot r+1}}).$ While this is applicable to a majority of real-world CFGs, it does not cover all of them. In general, there is no known constant bound on the pathwidth of CFGs and thus~\cite{DBLP:journals/pacmpl/ConradoGL23}'s algorithm has a much higher running time in the theoretical worst case.

We now show that our algorithm above significantly improves the time complexity, i.e.~the dependence on $r,$ for the spill-free register allocation problem, as well. 

\begin{theorem} \label{thm:free}
	The algorithm of Section~\ref{sec:algo} can be directly applied to spill-free register allocation, i.e.~the case where the cost $c(e, f)$ is zero when $f$ does not map anything to $\perp$ and $+\infty$ otherwise, and solves the problem in time $O(|G| \cdot r^{\textcolor{red}{5 \cdot r + 5}}).$
\end{theorem}
\begin{proof}
	The proof is exactly the same as that of Theorem~\ref{thm:main}, with one additional observation as follows: If we have $|L(a)| > r$ for some vertex $a$ of the control-flow graph, i.e.~if there are more than $r$ variables alive at the same time, then the answer to spill-free register allocation is ``no''. This is because the vertices in $L(a)$ form a clique in $\mathbb I.$ Thus, we only have to consider the case where $|L(a)| \leq r$ for every $a.$ Therefore, in the analysis of the proof of Theorem~\ref{thm:main}, when we consider the variables that are alive at $k \leq 5$ vertices, we can be sure that there are at most $k \cdot r$ such variables. Hence, our total runtime is $O(|G| \cdot r^5 \cdot r^{5 \cdot r}) = O(|G| \cdot r^{5 \cdot r + 5}).$
\end{proof}

\para{Further Optimization} There is also a simple optimization which can improve the performance of our algorithm in practice. For the spill-free register allocation problem, we do not have any register preferences. Thus, renaming and permuting registers does not invalidate a valid assignment. Similarly, if an assignment is invalid, renaming registers cannot fix the conflicts. Thus, register assignments that are the same modulo register renaming are the same to us. This means we only need to store one representative from each equivalence class as the canonical representation. This significantly reduces the number of dynamic programming values that our algorithm has to compute. %One natural choice is to view these register assignment functions as tuples of registers, where the indices correspond to variables ordered by some global variable ordering. The representative would be the lexicographically minimal one. We denote this canonical representation of $f$ by $\overline{f}$, and say that $f_1$ and $f_2$ are equivalent up to renaming by $f_1 \simeq f_2$.

\para{Mixing with Other Compiler Optimizations} Our approach for register allocation is robust with respect to other compiler optimizations that do not introduce/remove program variables and thus keep the lifetime intersections between variables intact. This is because we can consider the CFG right before such optimizations as the input to our register allocator. There are also many common optimizations, such as lifetime-optimal speculative partial redundancy elimination (LOSPRE)~\cite{DBLP:conf/scopes/Krause21}, which change the variables but keep the program structured. In such cases, our approach can work on the resulting CFG after these optimizations are applied. However, if the compiler applies an optimization which both introduces new variables and changes the CFG to a non-structured graph, then our approach cannot be combined with it.
 
\section{LOSPRE} \label{sec:lospre}

Redundancy elimination (RE), i.e.~avoiding repeated and unnecessary computations of the same expression, has been a goal of optimizing compilers since their early days. Put simply, if the same expression $e$ is used in several different locations in a program, it might be beneficial to compute $e$ once, store it in a temporary variable, and then use it whenever the program reaches any of the locations that need $e.$ One of the first formalizations of this problem was provided in 1970 as Global Common Subexpression Elimination (GCSE)~\cite{cocke1970GCSE}. Later approaches considered removing redundancies that appear only in a subset of paths of the control-flow graph, leading to Partial Redundancy Elimination (PRE)~\cite{morel1979partial}. An enhancement to PRE, introduced by Lazy Code-Motion (LCM)~\cite{knoop1992lazycodemotion}, focuses on achieving 
lifetime optimality by minimizing the lifetimes of the temporary variables it introduces. This is also helpful for reducing register pressure. Another classical improvement is that of Speculative PRE (SPRE)~\cite{cai2003SPRE,gupta1998profile}, which selects the path for adding computations based on profiling 
information with the goal of maximizing the benefits of PRE. Putting the ideas of LCM and SPRE together leads to Lifetime-Optimal SPRE (LOSPRE), which is currently the most expressive approach to redundancy elimination and subsumes all other methods mentioned above. 

Now with the CFG $G=\{V,E\}$, we formally define the problem as follows:

\begin{itemize}
    \item \emph{Use set.} Consider an expression $e.$ We define the use set $U$ of $e$ as the set of all nodes of the CFG in which the expression $e$ is computed.

    \item \emph{Live set.} Our goal is to precompute the expression $e$ at a few points, save the result in a temporary variable \texttt{temp}, and then use   \texttt{temp} in place of $e$ in every node of $U.$ We denote the lifetime of the variable $\texttt{temp}$ by $L$ and call it our live set. 

    \item \emph{Invalidating set.} We say a node $v$ of the CFG invalidates $e$ if the statement at $v$ changes the value of $e.$ For example, if $e = \texttt{a+b},$ then the statement $\texttt{a = 0}$ invalidates $e.$ We denote the set of all invalidating nodes by $I.$ These nodes play a crucial role in LOSPRE since they force us to update the value saved in $\texttt{temp}$ by recomputing $e.$ We assume that the entry and exit nodes are invalidating since LOSPRE is an intraprocedural analysis that has no information about the program's execution before or after the current function.

    \item \emph{Calculating set.} Given the sets $U, L$ and $I$ above, we have to make sure the value of our temporary variable $\texttt{temp}$ is correct at every node in $U \cup L.$ Thus, for every edge $(x, y) \in E$ of the CFG where $x \not \in L$ and $y \in U \cup L,$ we have to insert a computation $\texttt{temp} = e$ between $x$ and $y.$ Similarly, if $x \in I,$ then the value stored at $\texttt{temp}$ becomes invalid after the execution of $x,$ requiring us to inject the same computation between $x$ and $y.$ Formally, the computation $\texttt{temp} = e$ has to be injected into the following set of edges of the CFG:
$$
C(U, L, I) = \{ (x, y) \in E ~\vert~ x \not\in L \setminus I ~\land~ y \in U \cup L \}.
$$

\end{itemize}

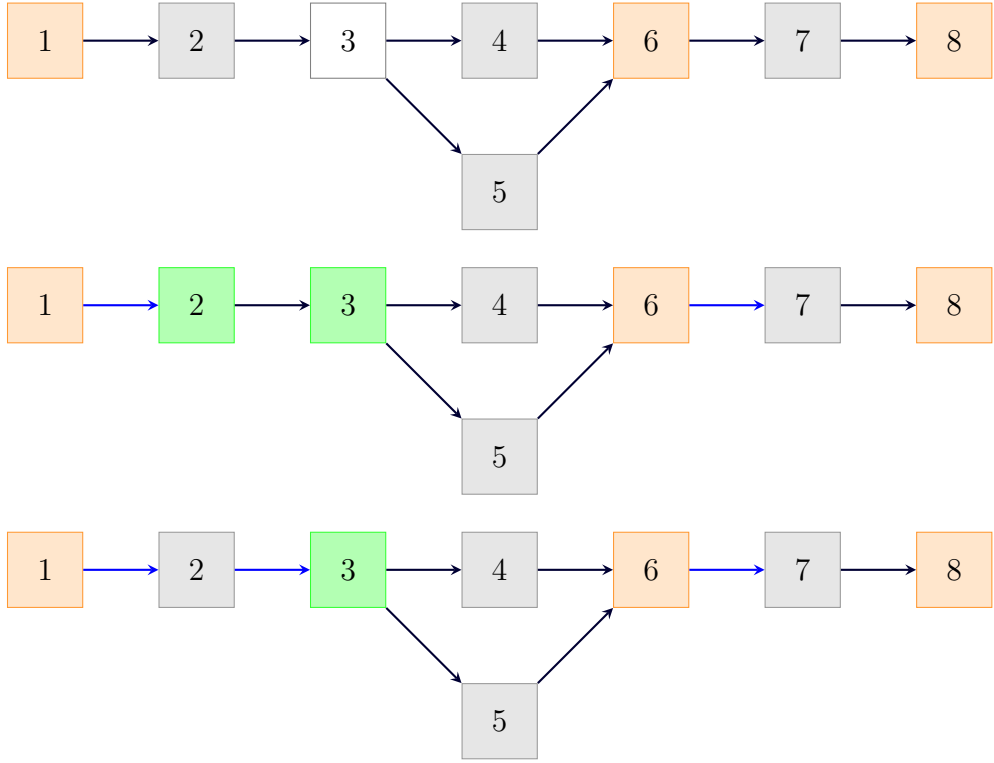
\begin{figure}
	\begin{center}
		\begin{tikzpicture}[scale=0.6]
			\node[redvertex] (1) {$1$};
			\node[grayvertex] (2)  [right =of 1] {$2$};
			\node[whitevertex] (3)  [right =of 2] {$3$};
			\node[grayvertex] (4)  [right =of 3] {$4$};
			\node[grayvertex] (5)  [below =of 4] {$5$};
			\node[redvertex] (6)  [right =of 4] {$6$};
			\node[grayvertex] (7)  [right =of 6] {$7$};
			\node[redvertex] (8)  [right =of 7] {$8$};
			\draw[arrow] (1) -- (2);
			\draw[arrow] (2) -- (3);
			\draw[arrow] (3) -- (4);
			\draw[arrow] (3) -- (5);
			\draw[arrow] (5) -- (6);
			\draw[arrow] (4) -- (6);
			\draw[arrow] (6) -- (7);
			\draw[arrow] (7) -- (8);
		\end{tikzpicture}
	\end{center}
	
	\begin{center}
		\begin{tikzpicture}[scale=0.6]
			\node[redvertex] (1) {$1$};
			\node[bluevertex] (2)  [right =of 1] {$2$};
			\node[bluevertex] (3)  [right =of 2] {$3$};
			\node[grayvertex] (4)  [right =of 3] {$4$};
			\node[grayvertex] (5)  [below =of 4] {$5$};
			\node[redvertex] (6)  [right =of 4] {$6$};
			\node[grayvertex] (7)  [right =of 6] {$7$};
			\node[redvertex] (8)  [right =of 7] {$8$};
			\draw[arrow,blue] (1) -- (2);
			\draw[arrow] (2) -- (3);
			\draw[arrow] (3) -- (4);
			\draw[arrow] (3) -- (5);
			\draw[arrow] (5) -- (6);
			\draw[arrow] (4) -- (6);
			\draw[arrow,blue] (6) -- (7);
			\draw[arrow] (7) -- (8);
		\end{tikzpicture}
	\end{center}
	
	\begin{center}
		\begin{tikzpicture}[scale=0.6]
			\node[redvertex] (1) {$1$};
			\node[grayvertex] (2)  [right =of 1] {$2$};
			\node[bluevertex] (3)  [right =of 2] {$3$};
			\node[grayvertex] (4)  [right =of 3] {$4$};
			\node[grayvertex] (5)  [below =of 4] {$5$};
			\node[redvertex] (6)  [right =of 4] {$6$};
			\node[grayvertex] (7)  [right =of 6] {$7$};
			\node[redvertex] (8)  [right =of 7] {$8$};
			\draw[arrow,blue] (1) -- (2);
			\draw[arrow, blue] (2) -- (3);
			\draw[arrow] (3) -- (4);
			\draw[arrow] (3) -- (5);
			\draw[arrow] (5) -- (6);
			\draw[arrow] (4) -- (6);
			\draw[arrow,blue] (6) -- (7);
			\draw[arrow] (7) -- (8);
		\end{tikzpicture}
	\end{center}
	\caption{An example of LOSPRE. The use set is shown in gray, the invalidating set in orange, and the life set in green. The edges of the calculation set are shown in blue.}
	\label{fig:ex}
\end{figure}

Figure~\ref{fig:ex} shows an example of LOSPRE. The top part of the figure is a CFG in which the use set of an expression $e$ is shown in gray, i.e.~we need the value of $e$ at vertices $U = \{2, 4, 5, 7\}.$ The invalidating set is shown in orange, i.e.~the vertices in $I = \{1, 6, 8\}$ invalidate $e.$ The middle and bottom parts each show one possible optimization. We show the lifetime of our temporary variable in green. 

In the middle part, the temporary variable is alive at $\{2, 3\}.$ Thus, the computation $\texttt{temp} = e$ has to be injected into the edge $(1, 2).$ We can then use $\texttt{temp}$ instead of $e$ in locations $2, 4$ and $5.$ However, we need to recompute $e$ in the edge $(6, 7).$ In this case our computation set is $\{(1, 2), (6, 7)\}.$ The edges in the computation set are shown in blue.

In the bottom part, the temporary variable is alive only at position $3.$ Thus, we first compute $e$ when passing through $(1, 2)$ so that we have its value at $2.$ We then recompute $e$ when going through $(2, 3)$ and save it at a temporary variable $\texttt{temp}.$ This temporary variable is then used in place of $e$ in $4$ and $5.$ This example shows a tradeoff in which fewer repetitions of the computation lead to a longer lifetime for the temporary variable, which increases register pressure and is undesirable for register allocation.

This time, there are two types of costs associated with the process above: (i)~injecting calculations into the edges in $C(U, L, I)$ and (ii)~keeping an extra variable $\texttt{temp}$ at every node in $L.$ These costs are dependent on the goals pursued by the compiler. For example, a compiler aiming to minimize code size will focus on (i). On the other hand, if our goal is to ease register pressure, we would want to minimize (ii). LOSPRE is an expressive framework in which these costs are modeled by two functions
$$
c: E \rightarrow K
$$ 
and 
$$
l: V \rightarrow K.
$$
where $K$ is a totally-ordered set with an addition operator, $c$ is a function that maps each edge to the cost of adding a computation of $e$ in that edge and $l$ is similarly a function that maps each vertex of the CFG to the cost of keeping the temporary variable $\texttt{temp}$ alive at that vertex.

\para{Formal Definition} Given a CFG $G=(V,E)$, a use set $U,$ an invalidating set $I$ and two cost functions $c: E \rightarrow K$ and $l: V \rightarrow K,$ the LOSPRE problem is to find a life set $L$ that minimizes the total cost $$\cost(G, U, I, L, c, l) = \sum_{e\in C(U,L,I)}c(e)+\sum_{v\in L} l(v).$$

\para{Our Algorithm}
We present a linear-time algorithm for LOSPRE using SPL decompositions. The input to our algorithm consists of a closed program $P,$ its control-flow graph $G=(V, E),$ a use set $U \subseteq V,$ an invalidating set $I \subseteq V$ and two cost functions $c: E \rightarrow K$ and $l: V \rightarrow K.$ Our goal is to find a life set $L \subseteq V$ that minimizes
$$\cost(G, U, I, L, c, l) = \sum_{e\in C(U,L,I)}c(e)+\sum_{v\in L} l(v).$$

\noindent\textbf{Step 1 (Initialization).} Our algorithm computes an SPL decomposition of $G = \cfg{P}$ by first parsing $P$ and then applying the homomorphism of the previous section. 

\noindent\textbf{Step 2 (Dynamic Programming).} Our algorithm proceeds with a bottom-up dynamic programming on the SPL decomposition. Note that each node $u$ of the SPL decomposition corresponds to an SPL subgraph $G_u$ of $G$, which is either an atomic SPL graph (when $u$ is a leaf) or obtained by applying one of the SPL operations to the graphs corresponding to the children of $u.$ See Figure~\ref{fig:decompo}. Let $\varGamma_u = \{S_u, T_u, B_u, C_u\}$ be the set of special vertices of $G_u.$ For every $X \subseteq \varGamma_u,$ we define a dynamic programming variable $\DP[u, X].$ Our goal is to compute this dynamic programming value such that
$$
\DP[u, X] = \min_{L \subseteq V_u ~\land~ L \cap \varGamma_u = X} \cost(G_u, U, I, L, c, l).
$$
Intuitively, we are considering a subproblem of the original LOSPRE in which the graph is limited to $G_u.$ Moreover, we only consider those solutions (life sets) $L$ for which $L \cap \varGamma_u = X.$ The value in $\DP[u, X]$ should then give us the minimum cost among all such solutions. Below, we present how our algorithm computes $\DP[u, X]$ for every vertex $u$ of the decomposition based on the $\DP[\cdot, \cdot]$ values at its children:
\begin{itemize}
	\item \emph{Atomic Nodes}: If $G_u$ is an atomic SPL graph, then the only vertices in $G_u$ are the four special vertices. Therefore, we must have $L=X.$ Our algorithm computes each $\DP[u, X]$ as:
	 $$\DP[u, X] = \cost(G_u, U, I, X, c, l) = \sum_{e\in C(U,X,I) \cap G_u} c(e)+\sum_{v\in X} l(v).$$
	 
	 \item \emph{Series Nodes}: Suppose $G_u = \oseries{G_v}{G_w}$ where $v$ and $w$ are the children of $u$ in the SPL decomposition. Let $X \subseteq \varGamma_u$ and $X_v \subseteq \varGamma_v$ be subsets of special vertices of $G_u$ and $G_v,$ respectively. We say that $X$ and $X_v$ are \emph{compatible} and write $X \compatible X_v$
	 if the following conditions are satisfied:
	 \begin{itemize}
	 \item 	$S_v \in X_v \Leftrightarrow S_u \in X;$
	 \item	$B_v \in X_v \Leftrightarrow B_u \in X;$
	 \item	$C_v \in X_v \Leftrightarrow C_u \in X.$
	 \end{itemize}
 	Intuitively, compatibility means that the subsets $X$ and $X_v$ make the same decisions about including vertices in the life set $L.$ Since $S_u = S_v,$ they should either both include it or both exclude it. Similarly, $B_u$ is obtained by merging $B_v$ and $B_w$. Therefore, the decisions made for $B_u$ and $B_v$ must match. The same applies to $C_u$ which is a merger of $C_v$ and $C_w.$ 
 	
 	Now consider $X_w \subseteq \Gamma_w.$ We say that $X_w$ and $X$ are compatible and write $X \compatible X_w$ if the following conditions are satisfied:
 	\begin{itemize}
 		\item 	$T_w \in X_w \Leftrightarrow T_u \in X;$
 		\item	$B_w \in X_w \Leftrightarrow B_u \in X;$
 		\item	$C_w \in X_w \Leftrightarrow C_u \in X.$
 	\end{itemize}
 	The intuition is the same as the previous case, except that we now have $T_u = T_w.$ Finally, we say that $X_v$ and $X_w$ are compatible and write $X_v \compatible X_w$ if 
 	\begin{itemize}
 		\item $T_v \in X_v \Leftrightarrow S_w \in X_w.$
 	\end{itemize}
 	This is because $T_v$ and $S_w$ are the same vertex of the CFG.
 	
 	In this step, our algorithm sets
 	\begin{equation*} 
 		\resizebox{\linewidth}{!}{
 		$\displaystyle \DP[u, X] = \min_{\makecell[c]{X \compatible X_v \\ X \compatible X_w \\ X_v \compatible X_w}} \DP[v, X_v] + \DP[w, X_w] - [T_v \in X_v] \cdot l(T_v) - [B_v \in X_v] \cdot l(B_v) - [C_v \in X_v] \cdot l(C_v).$}
 	\end{equation*}
 	This is because every edge in $G_u$ appears in either $G_v$ or $G_w$ but not both. Thus, the cost of the edges would simply be the sum of their costs in the two subgraphs. However, when it comes to vertices, $T_v$ and $S_w$ are merged, as are $B_v$ and $B_w,$ and $C_v$ and $C_w.$ Hence, we have to make sure we do not double count the cost of liveness for these vertices. Since this cost is counted in both $\DP$ values at the children, we should subtract it.
 	
 	\item \emph{Parallel Nodes:} We can handle parallel nodes in the same manner as series nodes, i.e.~finding compatible masks at both children and ensuring that there is no double-counting of the costs of vertices. To be more precise, let $G_u = \oparal{G_v}{G_w}.$ The compatibility conditions we have to check are as follows:
 \begin{equation*}
 	\resizebox{\linewidth}{!}{
 	$\begin{matrix}
 		X \compatible X_v \Leftrightarrow \\ \left( S_u \in X \Leftrightarrow S_v \in X_v ~\land~ T_u \in X \Leftrightarrow T_v \in X_v ~\land~ B_u \in X \Leftrightarrow B_v \in X_v ~\land~ C_u \in X \Leftrightarrow C_v \in X_v \right);\\
 		X \compatible X_w \Leftrightarrow \\ \left( S_u \in X \Leftrightarrow S_w \in X_w ~\land~ T_u \in X \Leftrightarrow T_w \in X_w ~\land~ B_u \in X \Leftrightarrow B_w \in X_w ~\land~ C_u \in X \Leftrightarrow C_w \in X_w \right);\\
 		X_v \compatible X_w \Leftrightarrow \\ \left( S_v \in X_v \Leftrightarrow S_w \in X_w ~\land~ T_v \in X_v \Leftrightarrow T_w \in X_w ~\land~ B_v \in X_v \Leftrightarrow B_w \in X_w ~\land~ C_v \in X_v \Leftrightarrow C_w \in X_w \right).
 	\end{matrix}$}
\end{equation*}
 	With the same argument as in the previous case, our algorithm sets
 \begin{equation*}\resizebox{\linewidth}{!}{$\displaystyle \DP[u, X] = \min_{\makecell[c]{X \compatible X_v \\ X \compatible X_w \\ X_v \compatible X_w}} \DP[v, X_v] + \DP[w, X_w] - [S_v \in X_v] \cdot l(S_v) -  [T_v \in X_v] \cdot l(T_v) - [B_v \in X_v] \cdot l(B_v) - [C_v \in X_v] \cdot l(C_v).$}
 	\end{equation*}
 	\item \emph{Loop Nodes:} Finally, we should handle the case where $G_u = \oloop{G_v}.$ This case is quite simple. By construction, in comparison to $G_v,$ the graph $G_u$ has four new vertices $$V_{\text{new}}=\{S_u, T_u, B_u, C_u\}$$ and five new edges $$E_{\text{new}} = \{ (S_u, S_v), (S_u, T_u), (T_v, S_u), (C_v, S_u), (B_v, T_u) \}.$$ The two graphs $G_u$ and $G_v$ do not share any special vertices, i.e.~$\varGamma_u \cap \varGamma_v = \emptyset.$ Moreover, for every edge $(x, y) \in E_{\text{new}}$ we can decide whether $(x, y)$ is in the calculation set solely based on $X$ and $X_v.$ This is because $x, y \in X \cup X_v.$ More specifically, $(x, y)$ is in the calculation set if and only if
 	$$
 	\varphi(X, X_v, x, y) := [x \not\in X \cup X_v \setminus I ~\land~ y \in U \cup X \cup X_v]
 	$$
 	
 	Thus, our algorithm sets:
 	$$
 	\DP[u, X] = \sum_{x \in V_{\text{new}} \cap X} l(x) + \min_{X_v \subseteq \varGamma_v} \DP[v, X_v] + \sum_{(x, y) \in E_{\text{new}}} 	\varphi(X, X_v, x, y) \cdot c(x, y).
 	$$
\end{itemize}

\noindent\textbf{Step 3 (Computing the Final Answer).} Let $r$ be the root of the SPL decomposition. By definition, we have $G_r = G.$ The algorithm outputs
$
	\min_{X \subseteq \varGamma_r} \DP[r, X]
$
as the minimum possible cost for the given LOSPRE input. This is because $G_r$ is the entire CFG $G$ and any solution $L$ will conform to exactly one of the different possible values of $X$ at $r.$ As is standard in dynamic programming approaches, one can reconstruct the optimal live set $L$ that leads to this minimal cost by retracing the steps of the algorithm and remembering which choices led to the optimal value at each step.

\begin{figure}
    \centering
\begin{tikzpicture}[scale=0.4]
  % Root
  \node[rect] (root) {$\oseries{}{}$};

  % Level 1
  \node[rect] (1-6) [below=of root] {$\oseries{}{}$};
  \node[rect] (6-8) [right=of 1-6]  {$\oseries{}{}$};
  \draw[thick] (root) -- (1-6);
  \draw[thick] (root) -- (6-8);

  % Level 2 (left subtree)
  \node[rect] (1-3) [below=of 1-6] {$\oseries{}{}$};
  \node[rect] (3--6) [right=of 1-3] {$\oparal{}{}$};
  \draw[thick] (1-6) -- (1-3);
  \draw[thick] (1-6) -- (3--6);

  % Level 2 (right subtree leaves)
  \node[rect] (6-7) [right=of 3--6] {$[6,7]$};
  \node[rect] (7-8) [right=of 6-7] {$[7,8]$};
  \draw[thick] (6-8) -- (6-7);
  \draw[thick] (6-8) -- (7-8);

  % Level 3 under [1,2,3]
  \node[rect] (1-2) [below=of 1-3] {$[1,2]$};
  \node[rect] (2-3) [right=of 1-2]  {$[2,3]$};
  \draw[thick] (1-3) -- (1-2);
  \draw[thick] (1-3) -- (2-3);

  % Level 3 under [3--6]
  \node[rect] (316) [right=of 2-3] {$\oseries{}{}$};
  \node[rect] (326) [right=of 316] {$\oseries{}{}$};
  \draw[thick] (3--6) -- (316);
  \draw[thick] (3--6) -- (326);

  % Level 4 under [3,4,6]
  \node[rect] (3-4) [below=of 316] {$[3,4]$};
  \node[rect] (4-6) [right=of 3-4] {$[4,6]$};
  \draw[thick] (316) -- (3-4);
  \draw[thick] (316) -- (4-6);

  % Level 4 under [3,5,6]
  \node[rect] (3-5) [right=of 4-6] {$[3,5]$};
  \node[rect] (5-6) [right=of 3-5] {$[5,6]$};
  \draw[thick] (326) -- (3-5);
  \draw[thick] (326) -- (5-6);
\end{tikzpicture}
    \caption{The SPL Decomposition of the Example in Figure~\ref{fig:ex}}
    \label{fig:splex}
\end{figure}
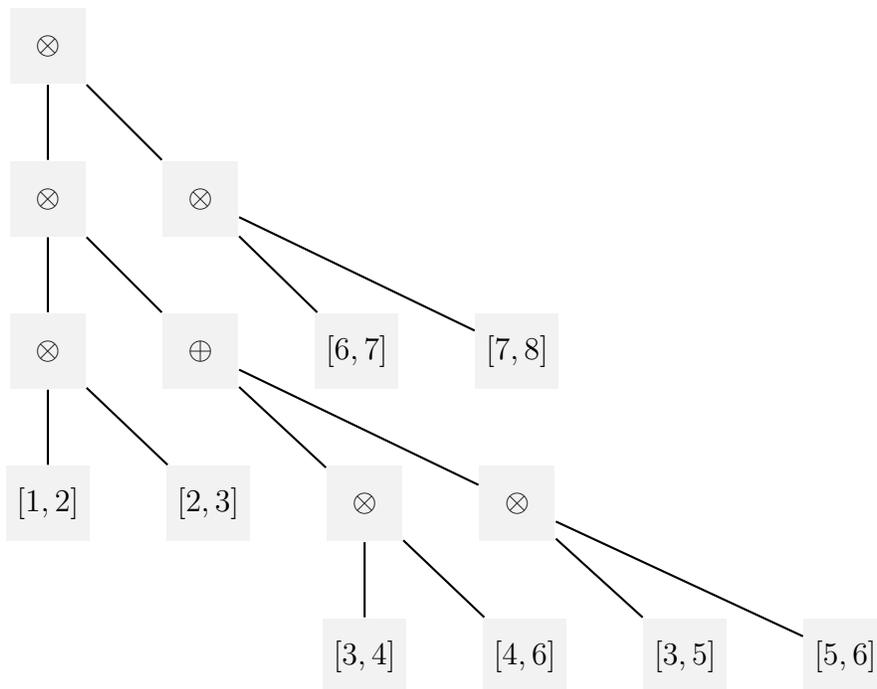

\noindent\textbf{Example.}  
We now illustrate in detail how our algorithm processes the CFG from
Figure~\ref{fig:ex}.  
For simplicity, we again assume a uniform edge-insertion cost $c(e)=1$
for all edges and a uniform liveness cost $l(v)=0.1$ for all vertices.
These uniform costs allow us to highlight the structural behavior of the DP
without obscuring the example with constant factors. Recall the relevant sets for the example:
\[
U=\{2,4,5,7\},\qquad I=\{1,6,8\}.
\]
Figure~\ref{fig:splex} shows an SPL decomposition of the CFG.  
Because the example contains no loops, the entire decomposition consists only of  
\emph{atomic}, \emph{series}, and \emph{parallel} nodes.  
For each node $u$ in this tree, we compute the DP table $\DP[u,X]$ for every
subset $X\subseteq \Gamma_u$ of its special vertices.

\medskip
\noindent\emph{Leaf nodes.}  
Each leaf of the decomposition corresponds to a simple two-vertex fragment
$[i,j]$ encoding a single CFG edge.  
Since the example has no loop blocks, each atomic subgraph has exactly two
special vertices $S_u=i$ and $T_u=j$; the break/continue vertices are irrelevant
in this non-loop setting and can be ignored.
Because $\Gamma_u=\{i,j\}$, we evaluate all four subsets
\[
X\in\bigl\{\emptyset,\,\{i\},\,\{j\},\,\{i,j\}\bigr\}.
\]
For each subset we evaluate
\[
\DP[u,X]
   \;=\;\sum_{e\in C(U,X,I)\cap G_u} c(e)
     \;+\;\sum_{v\in X} l(v).
\]

At a leaf $[i,j]$, the only edge is $(i,j)$, so the only question is whether this
edge belongs to the calculation set.  
This is determined by whether $i\notin X\setminus I$ and $j\in U\cup X$.  
If the update must be injected at $(i,j)$ then we add cost $1$;  
if the vertex $i$ or $j$ is included in $X$, we add the appropriate liveness
cost.

As a concrete example, consider the fragment $[3,4]$.  
Since $4\in U$, any computation of the expression must be correct at node~4.
Thus, unless $3\in X$, the edge $(3,4)$ must be present in the calculation set
and contributes cost~$1$.  
We obtain:
\[
\DP([3,4],\emptyset)=1,\qquad
\DP([3,4],\{3\})=0.1,\]
\[
\DP([3,4],\{4\})=1.1,\qquad
\DP([3,4],\{3,4\})=0.2.
\]

Carrying this out for all leaves produces a full table of $4$ entries per leaf.
Even though each leaf is a tiny graph, the masks on its interface (the special
vertices) are essential because they allow the DP to distinguish when the value
of $e$ must already be live at $S_u$ or $T_u$, which in turn influences whether
recomputation is required.

\medskip
\noindent\emph{Series nodes.}  
A series node corresponds to concatenating two SPL subgraphs so that the
terminate vertex of the left child coincides with the start vertex of the right
child.  
Let us examine the series node combining $[1,2]$ and $[2,3]$.  
Here, the merged vertex is
\[
T_{[1,2]} = S_{[2,3]} = 2,
\]
so the interface of the parent node is
\[
\Gamma_u = \{1,3,B_u,C_u\}.
\]
Again, $B_u$ and $C_u$ play no role, but are part of the uniform interface
structure.

For each possible $X\subseteq\Gamma_u$, the DP considers all masks
$X_v\subseteq\Gamma_v$ and $X_w\subseteq\Gamma_w$ from the two children that
respect the compatibility constraints.  
Compatibility enforces that the merged vertex $2$ is consistently included or
excluded across all three masks.  
Thus, even if the parent mask $X$ does not explicitly mention vertex~2,
every compatible pair $(X_v,X_w)$ implicitly accounts for whether the merged
vertex is live at child interfaces.

Let us illustrate this with a detailed evaluation.  
Suppose we compute $\DP([1,3],X)$ for $X=\{3\}$, meaning that the temporary value
is required to be live at vertex~3 but not at~1.  
To satisfy this, the DP enumerates compatible pairs $(X_v,X_w)$ for the children
$[1,2]$ and $[2,3]$.

Concretely, two such compatible pairs arise:

\begin{itemize}
  \item $X_v=\{2\}$ and $X_w=\{2,3\}$.  
        This choice ensures that the merged vertex~2 is live in both children.
        The combined cost must subtract the double-counted liveness of vertex~2.
  \item $X_v=\emptyset$ and $X_w=\{3\}$.  
        This corresponds to the case where vertex~2 is not required to be live.
\end{itemize}

Evaluating both pairs gives:
\[
\begin{small}
\begin{aligned}
\DP([1,3],\{3\})
  &= \min\Bigl\{
       \DP([1,2],\{2\}) + \DP([2,3],\{2,3\}) - l(2),
       \;\DP([1,2],\emptyset) + \DP([2,3],\{3\})
     \Bigr\} \\
  &= \min\{1.1 + 0.2 - 0.1,\;\; 1 + 1.1\}
   = 1.2.
\end{aligned}
\end{small}
\]

In a full implementation, the DP would evaluate all $2^{|\Gamma_u|}=4$ masks at
this node, but the example above illustrates how compatibility propagates
information about intermediate liveness and how vertex liveness costs are
deduplicated.

\medskip
\noindent\emph{Parallel nodes.}  
Parallel nodes model branching in the CFG.  
In the example, the two subgraphs $[3,4,6]$ and $[3,5,6]$ form parallel
branches that rejoin at vertex~6.  
The special vertices for the combined node are therefore $\{3,6\}$, and any
mask must agree on which of these vertices are live across both children.

For instance, consider $X=\{3\}$.  
The DP simply adds the values of the two child DPs but subtracts the liveness
cost of vertex~3 if it was counted twice:
\[
\DP([3,4,5,6],\{3\})
  = \DP([3,4,6],\{3\})
    + \DP([3,5,6],\{3\})
    - l(3).
\]
Similar reasoning applies to the other masks:
\[
X\in\{\emptyset,\{3\},\{6\},\{3,6\}\}.
\]
Each of these requires consistent inclusion/exclusion across both children and a
single correction term for double-counted liveness.

Intuitively, series nodes propagate information \emph{along} control flow, while
parallel nodes propagate information \emph{across} branches, ensuring that
branch-sensitive costs are treated correctly.

\medskip
\noindent\emph{Propagation toward the root.}  
Once all internal nodes have been processed, the DP reaches the root
representing the full CFG.  
The root has special vertices $\Gamma_r=\{1,8\}$, marking the entry and exit of
the function.  
We evaluate all four masks:
\[
X\in\bigl\{\emptyset,\{1\},\{8\},\{1,8\}\bigr\}.
\]

The minimum is attained at $X=\emptyset$ with value
\[
\DP(r,\emptyset)=2.2.
\]
This states that the temporary value is not needed at either the entry or exit
vertex, and the minimal number of inserted computations (plus liveness costs)
required to satisfy all uses is $2.2$ under our cost model.

Tracing back the choices of compatible masks at each decomposition step reveals
the corresponding optimal life set:
\[
L=\{2,3\}.
\]
Thus, the temporary variable must be live only at vertices $2$ and $3$ in order
to achieve the minimum cost.  
Every other node either recomputes the expression or does not require its
value, exactly as determined by the DP.

This expanded walkthrough demonstrates how the SPL decomposition guides the DP
by exposing small, structured interfaces and ensuring that global optimality can
be achieved through local compatibility constraints.

\begin{theorem} Given a LOSPRE instance consisting of a closed structured program $P,$ its control-flow graph $G$ with $n$ vertices, a use set $U,$ an invalidating set $I$ and two cost functions $c: E \rightarrow K$ and $l: V \rightarrow K,$ the algorithm above solves the LOSPRE problem in $O(n)$ and outputs 
	$$
	\min_{L} \cost(G, U, I, L, c, l) ~~~~~~ \text{ and } ~~~~~~
	\argmin_{L} \cost(G, U, I, L, c, l).
	$$
\end{theorem}
\begin{proof}
	Correctness is already argued above. Thus, we focus on the runtime analysis. The SPL decomposition has $O(n)$ vertices and can be computed in $O(n)$ . At each vertex $u$ of the decomposition, we have $2^4 = 16 = O(1)$ different possible values for $X.$ The computations in atomic node are over graphs with only four vertices and thus take $O(1)$ time. In series and parallel node, we have at most two compatible $X_v$'s for each $X.$ This is because inclusion or exclusion of the vertices $S_v, B_v$ and $C_v$ in $X_v$ is uniquely determined by $X$ and only $T_v$ remains to be chosen. Similarly, for every fixed $X, X_v,$ there is a unique $X_w.$ Thus, computing each $\DP[u, X]$ in this step takes $O(1)$ time. In loop node, every $X$ induces a unique $X_v$ and a unique $X_w.$ Hence, this step takes $O(1)$ time to compute each $\DP[u, x]$ value. In step 3, we try $2^4 = O(1)$ different $X_v$'s for each $X.$ Thus, the total runtime of Step 2 is $O(n).$ Finally, Step 3 takes the maximum of $2^4 = O(1)$ values.
\end{proof}

\section{Experimental Results} \label{sec:exp}

In this section, we provide extensive experimental results comparing our algorithms for spill-free register allocation and LOSPRE with previous approaches based on treewidth~\cite{DBLP:conf/cc/Krause13,DBLP:conf/scopes/Krause21}.

\para{Implementation} We implemented our approaches in \texttt{C++} and integrated them with the Small Device C Compiler (SDCC)~\cite{sdcc2}. SDCC already includes a heavily optimized variant of the algorithms from~\cite{THORUP1998159,DBLP:journals/dam/KrauseLS20, DBLP:conf/cc/Krause13} for finding tree decompositions and the treewidth-based algorithms for the two compiler optimization tasks. Despite our approach being perfectly parallelizable, we did not use parallelization in our experiments in order to provide a fair comparison with the available implementations of previous methods, which are not parallel. 

\para{Machine}  The results were obtained on a machine running Oracle Linux (ARM 64-bit), equipped with 1 core CPU of Apple M2 and 4GB of RAM.

\para{Benchmarks} We followed the setup described in~\cite{DBLP:journals/pacmpl/ConradoGL23}. We used the SDCC regression test suite for HC08 as our benchmark set. These benchmarks consist of embedded programs that are designed to run on systems with limited resources, making compiler optimization a critical performance bottleneck for them. The functions within these benchmarks have control flow graphs (CFGs) ranging from 1 to 800 vertices (lines of code), with an average size of 15.7 vertices. Figure~\ref{fig:cfgSize} presents a histogram of function sizes. We run register allocation once for each test case and LOSPRE once for each expression, so LOSPRE would have more instances.

\begin{figure}
\centering
	\includegraphics[width=0.8\linewidth]{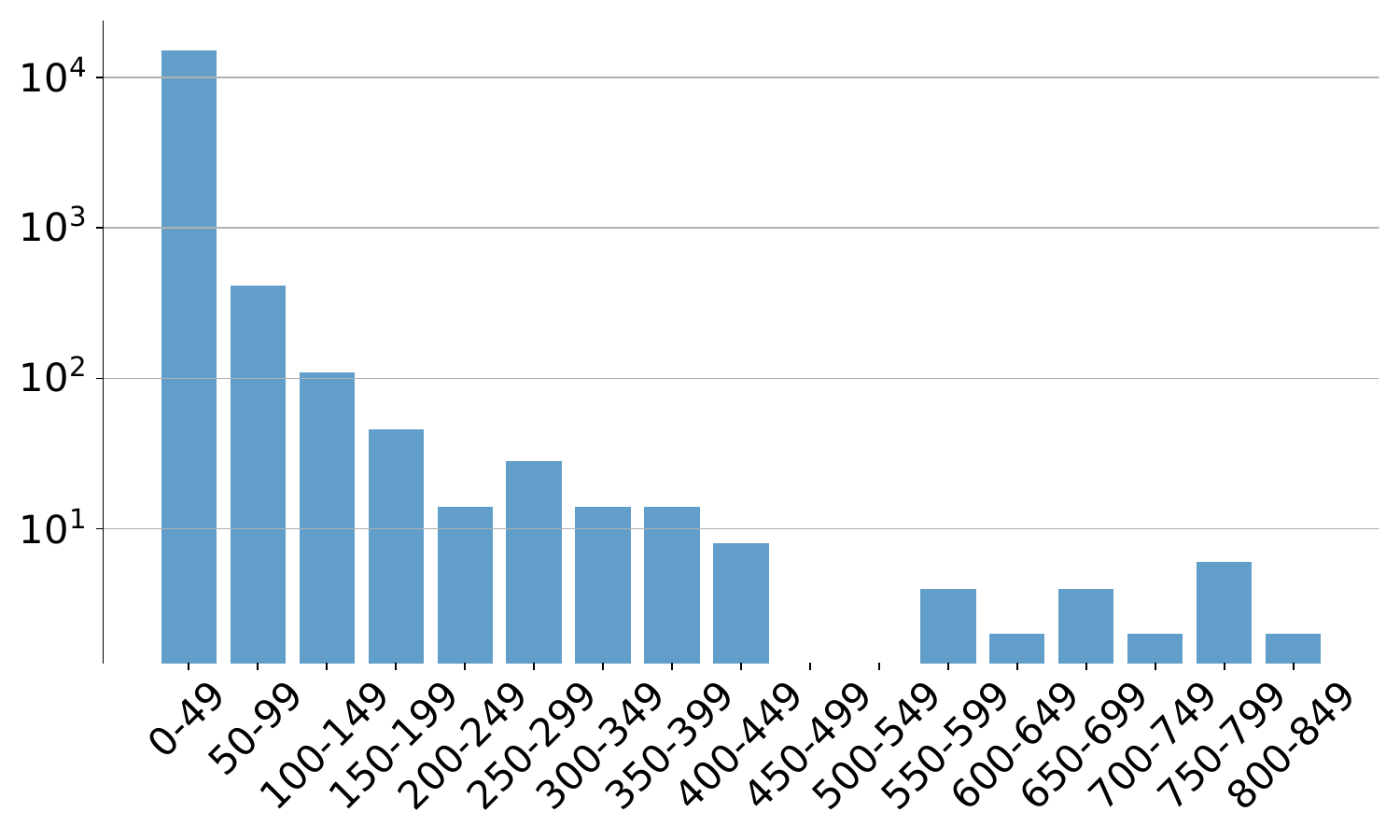}
	\caption{Histogram of the number of CFG vertices (lines of code) in each benchmark. The $x$ axis is the CFG size and the $y$ axis is the number of instances. The $y$ axis is in logarithmic scale.}
	\label{fig:cfgSize}
\end{figure}

\begin{figure}
\centering
	\includegraphics[width=0.6\linewidth]{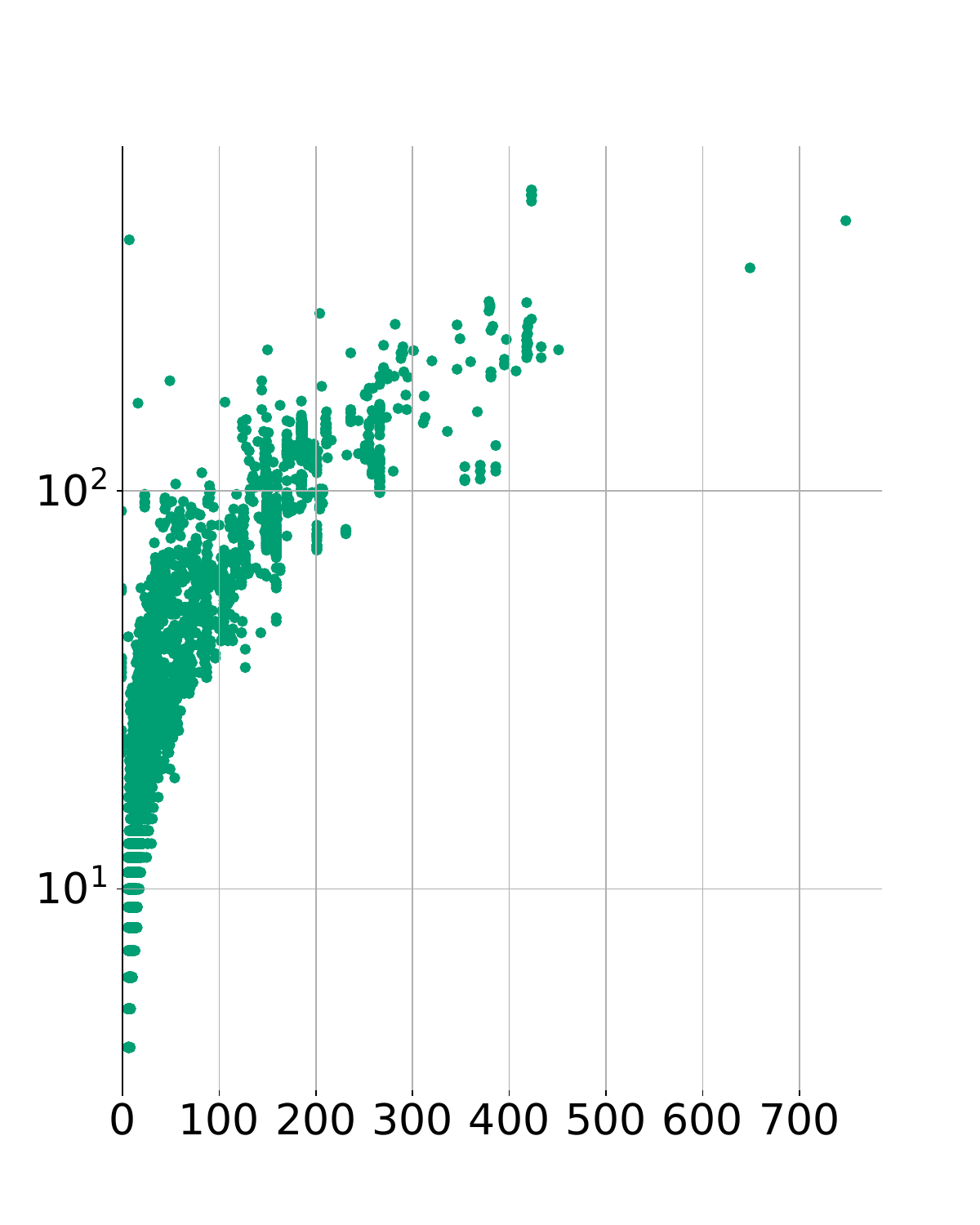}
	\caption{The runtime needed for computing grammatical decompositions of CFGs by parsing the programs. Each dot corresponds to one instance. The $x$ axis is the size of the CFG and the $y$ axis is the runtime in microseconds. The $y$ axis is in logarithmic scale.}
	\label{fig:spl-decom-time}
\end{figure}

\begin{figure}
\centering
	\includegraphics[width=0.8\linewidth]{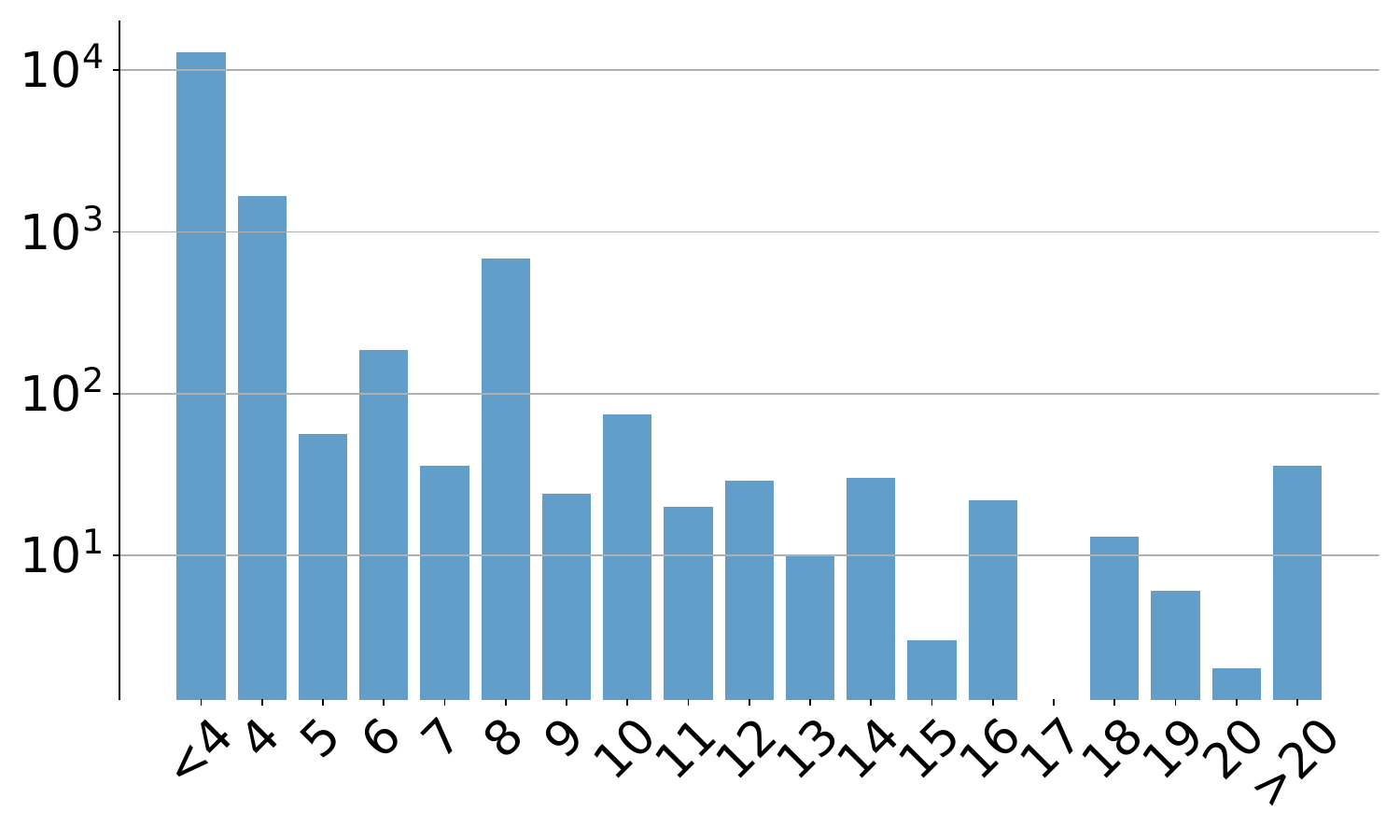}
	
	\caption{Histogram of the minimum number of registers required for spill-free allocation. The $x$ axis is number of registers and the $y$ axis is the number of instances requiring that many registers. The $y$ axis is in logarithmic scale.}
	\label{fig:register-frequency}
\end{figure}

\para{Runtimes for Computing the Grammatical Decomposition} All three compiler optimization tasks rely on grammatical decompositions of control flow graphs (CFGs) as SPL graphs. As mentioned above, such decompositions can be computed in linear time from the original CFG. Figure~\ref{fig:spl-decom-time} illustrates the runtime required for computing grammatical decompositions for each of our benchmarks. The average runtime was 13.8 microseconds, with a maximum of 570 microseconds. Therefore, grammatical decompositions can be computed extremely efficiently, and the time spent obtaining them does not significantly contribute to the total compile time.

\begin{figure*}

	\subfloat[Register Allocation]{
		\includegraphics[clip,width=.5\textwidth]{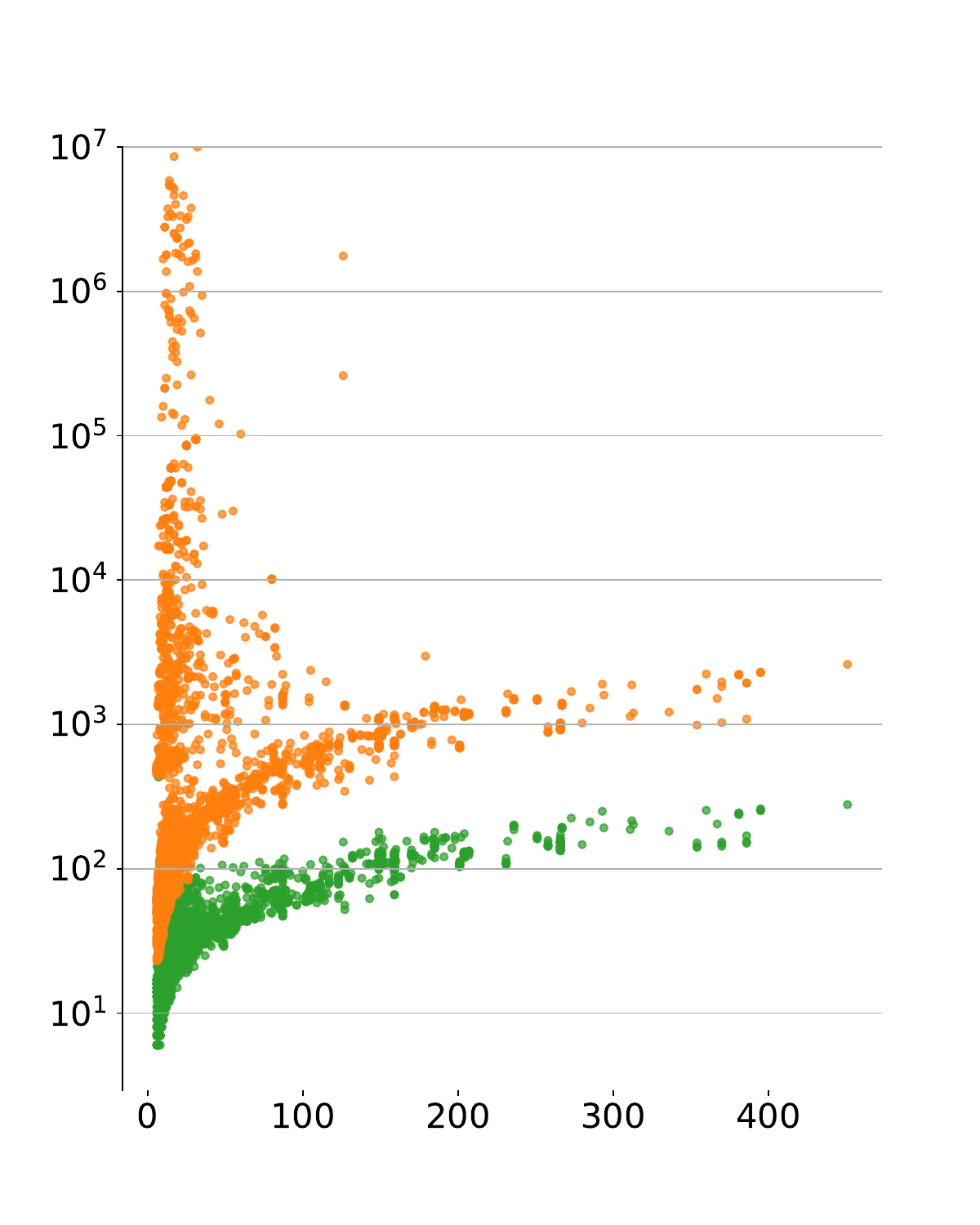} % Increased width to 0.45
	}%
	\subfloat[LOSPRE]{
		\includegraphics[clip,width=.5\textwidth]{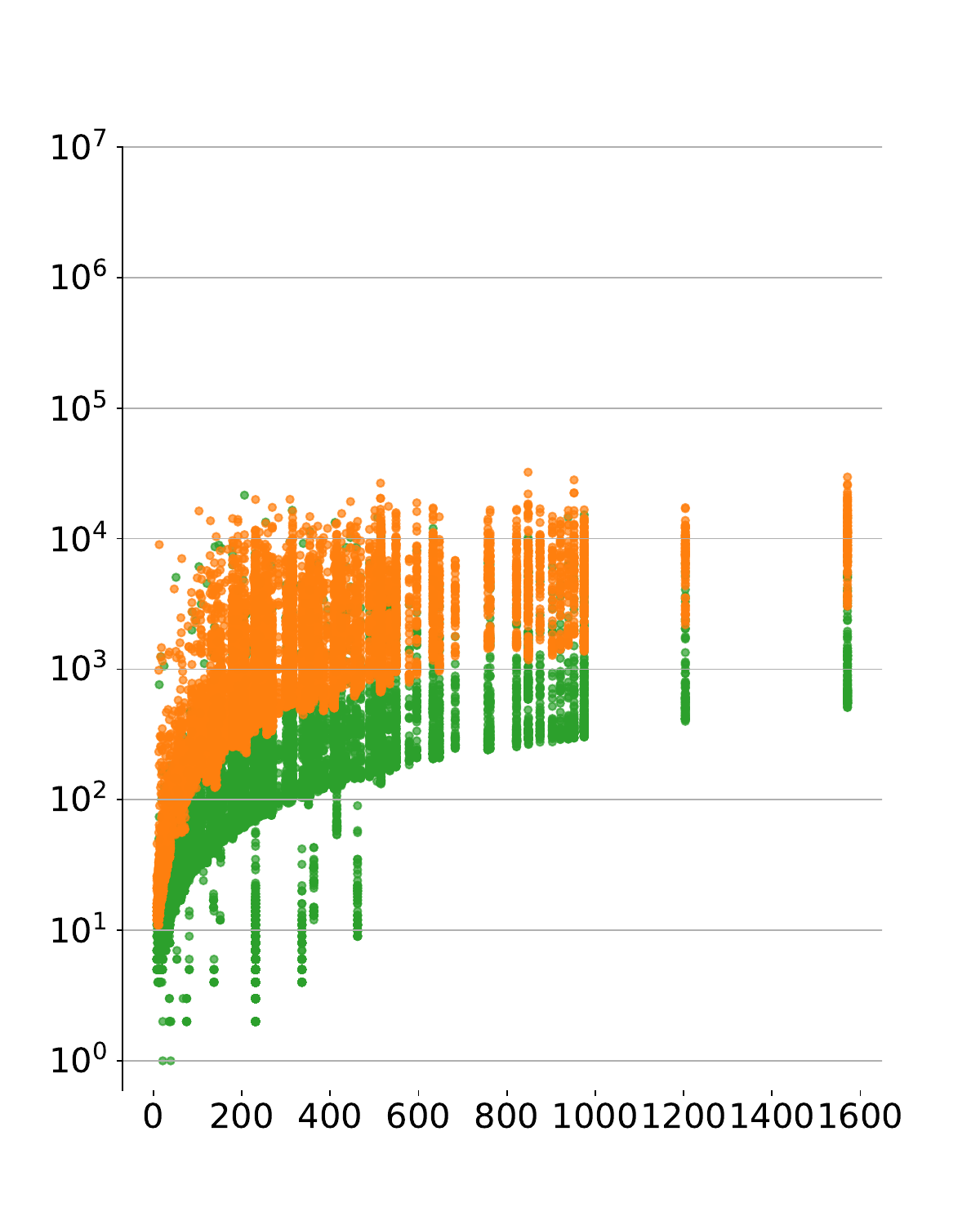} % Increased width to 0.45
	}%
	\caption{Runtime comparison of the treewidth-based algorithm (orange) and our approach (green) for (a) Register Allocation, (b) LOSPRE. The $x$ axis represents the number of vertices in the CFG, and the $y$ axis represents the time in (a,b) microseconds.}
	\label{fig:runtime}
\end{figure*}

\para{Register Allocation}
Given an input program, our objective is to determine the smallest number \( r \) of registers required for spill-free allocation. As previously mentioned, spill-free register allocation is a specific case of register allocation characterized by minimum cost, where the cost is zero if there is no spilling and infinite otherwise. We selected this problem for our experimental evaluation for two reasons: (i) most previous works in the literature focus on this variant, and (ii) there is generally no standard method for selecting the cost function \( c \); each compiler defines this function differently based on its own context and use cases, often relying on dynamic analysis and profiling. In contrast, spill-free allocation is well-defined and consistent across all compilers.

 Our approach successfully handled all input instances within the prescribed time and memory limits, either finding the optimal number of registers needed for spill-free allocation or reporting that more than 20 registers are required. Figure~\ref{fig:register-frequency} shows a histogram of the number of required registers. In contrast, the treewidth-based approach of~\cite{DBLP:conf/cc/Krause13} failed in 554 instances, including all instances requiring more than 8 registers.

Figure~\ref{fig:runtime} (a) shows a comparison of the runtimes of our algorithm vs the treewidth-based approach of~\cite{DBLP:conf/cc/Krause13}, when we set $r \leq 20.$ The average runtimes were 3.87 microseconds for our approach and  1,191,284  microseconds for~~\cite{DBLP:conf/cc/Krause13}. These averages are excluding the instances over which the previous methods failed. The runtimes were dominated by 704 instances for the treewidth-based approach, presumably due to high treewidth. Excluding these outlier instances, the average runtime was 372.34 microseconds for~\cite{DBLP:conf/cc/Krause13} .

\begin{figure}
	\subfloat[Pathwidth]{
		\includegraphics[clip,width=.5\textwidth]{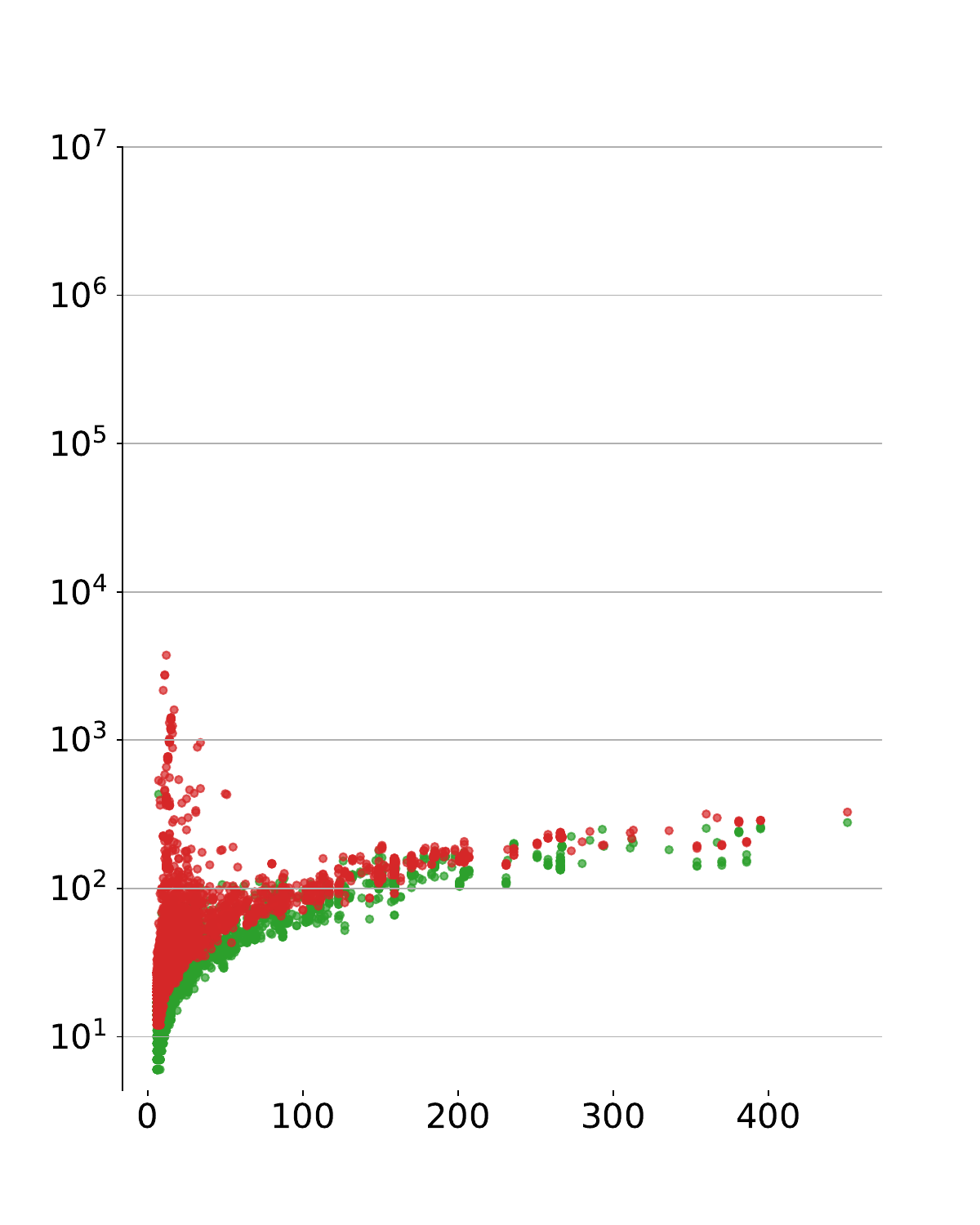} % Increased width to 0.45
	}%
	\subfloat[Graph Coloring]{
		\includegraphics[clip,width=.5\textwidth]{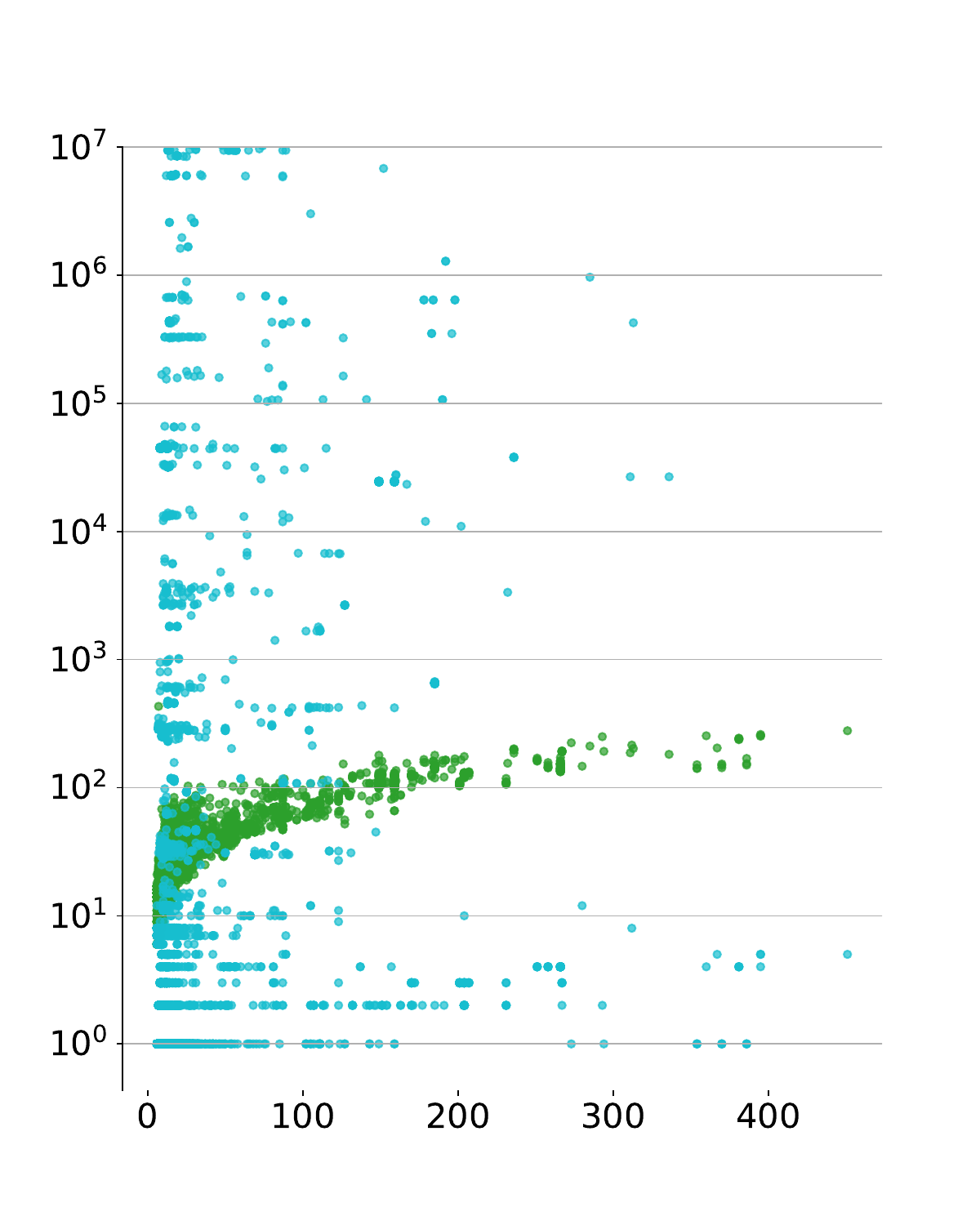} % Increased width to 0.45
	}%
	\caption{Runtime comparison of our approach with other algorithms for register allocation. (a) is based on \cite{DBLP:journals/pacmpl/ConradoGL23}, (b) is based on \cite{DBLP:conf/pldi/Chaitin82}. The x axis is the number of vertices in the CFG, and the y axis is the time in microseconds. The y axis is on a logarithmic scale.}
	\label{fig:runtimera}
\end{figure}

Figure~\ref{fig:runtimera} shows a comparison of runtime of our approach vs the pathwidth-based approach of~\cite{DBLP:journals/pacmpl/ConradoGL23} and the classical graph coloring approach of ~\cite{DBLP:conf/pldi/Chaitin82}. When we set $r \leq 20.$ The average runtimes were 21,544  microseconds for~~\cite{DBLP:journals/pacmpl/ConradoGL23} and for our approach, the average runtimes were 3.87 microseconds. These averages are excluding the instances over which the previous methods failed. 

We observe that Chaitin's graph coloring method, which is the only classical non-parameterized approach that uses the optimal number of registers, is highly unscalable. Given a time limit of 1 minute, it could handle only 6,042 benchmarks in our suite with an average runtime of 153,602 microseconds. Notably, this did not include any benchmark which required more than $8$ registers. Graph coloring timed out on all such benchmarks.

\begin{figure}
	\centering
		\includegraphics[width=\linewidth]{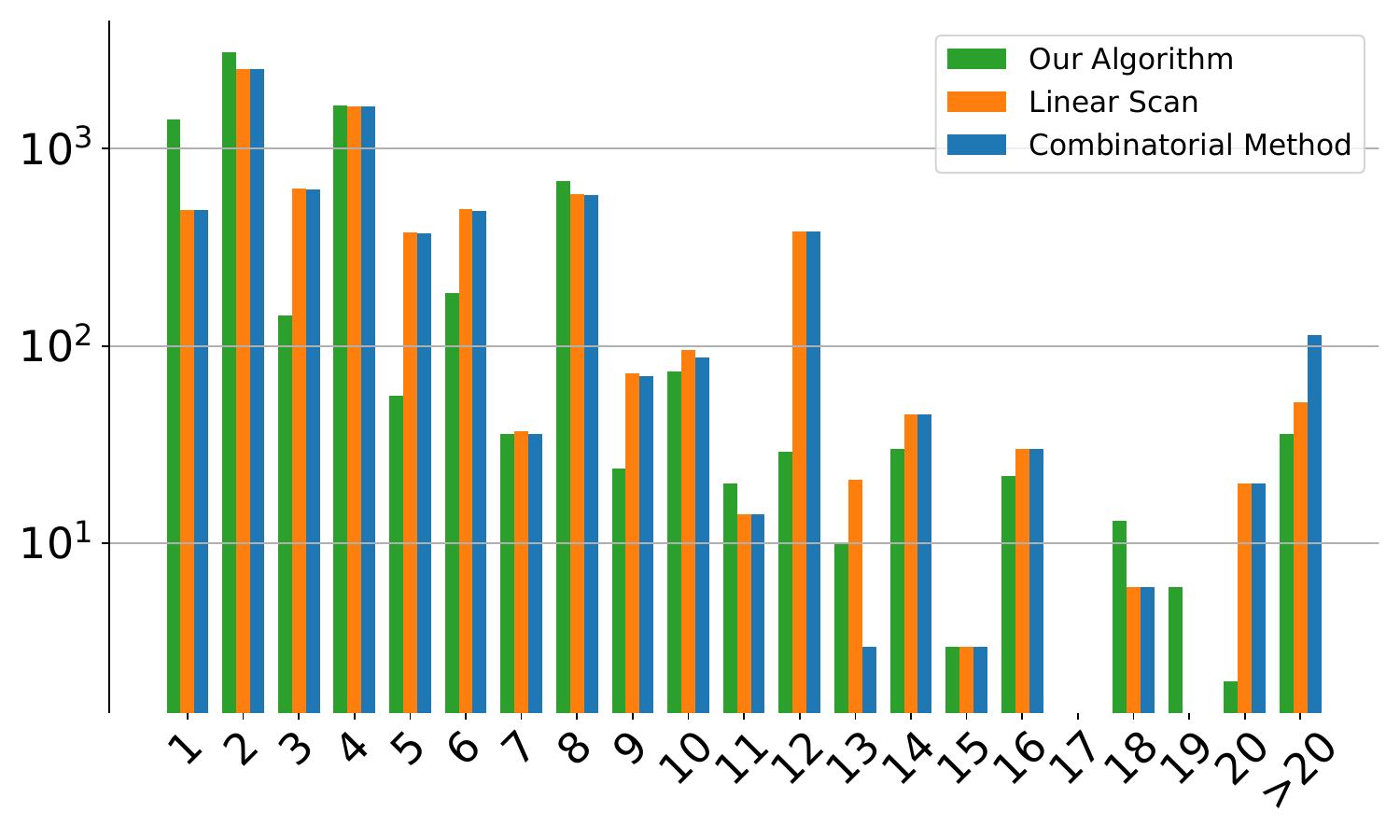}
	\caption{Histogram of the number of registers used by our algorithm, linear scan~\cite{poletto1999linear} and the combinatorial methods of~\cite{lozano2019combinatorial}. The $x$ axis is the number of registers used and the $y$ axis is the number of instances in logarithmic scale. Our algorithm always uses the minimum possible number of registers. This figure excludes trivial instances that needed no registers at all.}
	\label{fig:register-needed}
\end{figure}

Unlike graph coloring, the two heuristic methods, namely linear scan~\cite{poletto1999linear} and the combinatorial approach of~\cite{lozano2019combinatorial}, are quite efficient in practice. For the former, the average runtime was less than one microsecond, with a maximum of 9 microseconds over all benchmarks. For the latter, the average was 15 microseconds and the maximum runtime was 1,363 microseconds. Thus, their runtimes are comparable to our approach and linear scan is even faster. However, they are not guaranteed to use the minimum required number of registers as they are only heuristics and do not come with optimality proofs. Since linear scan and the combinatorial method are not exact, we compared their output to that of our approach in terms of the number of used registers. Figure~\ref{fig:register-needed} shows a histogram of the number of used registers by each approach. We observe that both heuristic methods often use more registers than our optimal solution. Indeed, linear scan finds the optimal number of registers 46 percent of the time whereas the combinatorial method of~\cite{lozano2019combinatorial} generates the optimal result in 45 percent of the instances.

\para{LOSPRE}
In our LOSPRE instances, the goal is to minimize the total number of computations in the resulting 3-address code. Thus, we use $K = \mathbb Z^2$ with lexicographic ordering. The cost assigned to each edge $(x, y)$ is $c(x, y) = (1, 0).$ We also enforce lifetime-optimality by assigning the cost $l(x) = (0, 1)$ to every vertex $x.$ As both approach can find the minimum cost, we only compare the time used.

Figures~\ref{fig:runtime} (b) provide runtime comparisons between \cite{DBLP:conf/scopes/Krause21} and our approach. On average, our algorithm takes 222.38 microseconds, while the treewidth-based approach of \cite{DBLP:conf/scopes/Krause21} has an average runtime of 1349.14 microseconds. The maximum runtime was 21,524 microseconds for our algorithm compared to 32,284 microseconds for \cite{DBLP:conf/scopes/Krause21}. Our algorithm significantly outperforms~\cite{DBLP:conf/scopes/Krause21} in the vast majority of benchmarks. We identified only 19 instances where our runtime exceeded 10,000 microseconds, whereas \cite{DBLP:conf/scopes/Krause21} takes more than 10,000 microseconds in 277 instances.

\para{Discussion} In all cases above, our algorithm outperforms the previous state-of-the-art. For register allocation, our approach is the first exact algorithm for spill-free register allocation that scales to realistic architectures with up to 20 registers, such as those in the \texttt{x86} family. Given the efficiency of our method, which achieves an average runtime of merely 4 microseconds per instance, we believe there is no longer a justification for using approximations or heuristics in spill-free register allocation. Despite its NP-hardness and theoretical hardness of approximation, our approach efficiently solves this problem for all practical instances. For LOSPRE, our approach is about 6 times faster than the treewidth-based algorithm.

\section{Broader Implications and Future Directions}

In this section, we expand on the broader conceptual implications of our work
and discuss how the SPL framework reveals new connections across compiler
optimization tasks. We further compare SPL decomposition with classical
treewidth-based techniques and outline several concrete directions for future
research. Our goal is to situate SPL decompositions within the larger ecosystem
of program analysis methods and highlight why this perspective is promising for
scaling to increasingly complex optimization pipelines.

\subsection{Connections Between Register Allocation and LOSPRE}

Although register allocation and LOSPRE pursue different objectives, both
problems exhibit a structural property that makes them particularly well-suited
for decomposition-based reasoning: \emph{locality of dependence}.
In both applications, decisions at a given point in the program depend on a
restricted set of variables or states that interact only through a small
boundary. This phenomenon is also reflected in their classical graph-theoretic
reductions: register allocation reduces to coloring portions of the interference
graph, whereas LOSPRE reduces to computing feasible placements of expression
evaluations under local data dependencies.

A central observation is that neither problem requires global knowledge of the
entire CFG at once. Instead, only a compact summary---the set of variables that
are live, active, or otherwise relevant at a boundary region---is needed to
propagate information across the decomposition. In the case of register
allocation, this boundary consists of live variables whose register assignments
must remain consistent across adjacent regions. For LOSPRE, the boundary
captures whether the temporary value representing the expression must be alive
before or after a region, and whether the region invalidates this value.

This locality behavior parallels that of classical width-based graph problems.
For example, solving a $k$-coloring problem on a CFG using an
SPL-style decomposition requires storing only the color assignments at special
boundary nodes. As long as the boundary size is small, the dynamic program
remains tractable, yielding an FPT algorithm with respect to the number of
colors. The same principle underlies many data-flow analyses such as available
expressions, reaching definitions, and liveness propagation. These analyses
typically maintain a finite abstract state at each boundary, making them strong
candidates for SPL-style reasoning.

The connection between register allocation and LOSPRE thus suggests that SPL
decompositions provide a unifying perspective for a wide class of optimization
problems. Tasks that share similar locality patterns may be amenable to
efficient algorithms that operate over the decomposition rather than the entire
CFG. Identifying such tasks and formalizing their boundary summaries is an
exciting avenue for expanding the reach of decomposition-based optimization.

\subsection{Comparison Between SPL Decomposition and Tree Decomposition}

Treewidth has long been a powerful parameter for analyzing CFGs, particularly
since structured programs have been shown to have treewidth at most
seven. However, tree decompositions operate on undirected graphs and treat all
edges symmetrically, which obscures the directional nature of control flow.
SPL decomposition, by contrast, is inherently directional and built from the
same grammatical constructs that generate structured programs. This alignment
with program syntax leads to several important advantages.

First, SPL decomposition yields consistently small cut sizes---never more than
four special vertices for structured programs. In contrast, tree decomposition
may require bags containing up to eight vertices even for well-structured
inputs. Since dynamic programming complexity is usually exponential in the cut
size, this reduction has a substantial impact on performance. Our experiments
support this intuition: smaller boundaries lead directly to smaller DP tables,
fewer compatibility checks, and faster per-node processing.

Second, the SPL construction faithfully preserves the control-flow semantics.
Series nodes reflect sequential composition; parallel nodes encode branching
and merging; loop nodes capture precisely the behavior of iterative constructs.
This semantic fidelity enables problem-specific optimizations, such as
propagating information in the direction of control flow, leveraging loop
structure, or simplifying local states based on dominance relationships.
Treewidth-based methods, lacking access to this structure, cannot exploit these
properties as directly.

Nevertheless, tree decomposition remains a robust and general tool. When
\texttt{goto} statements or other unstructured control-flow constructs are
introduced, the CFG may no longer be decomposable using the SPL grammar.
Treewidth-based methods then provide a fallback: even though the treewidth may
increase, empirical studies suggest that many real-world programs still exhibit
relatively small treewidth in practice. Understanding why unstructured
constructs rarely cause large blow-ups in treewidth, and whether such behavior
is fundamental or incidental, remains an intriguing open question.

\subsection{Future Research Directions}

Our findings open several promising directions for future exploration.

\paragraph{Extending SPL to Broader Optimization Tasks}
Many classical analyses---such as available expressions, reaching definitions,
constant propagation, and partial dead-code elimination---exhibit boundary
summaries similar to those for register allocation and LOSPRE. These
summaries often depend on a small number of variables or expression states,
suggesting that SPL decomposition could support efficient dynamic-programming
solutions. Formalizing these summaries and developing general frameworks for
SPL-based data-flow analyses could significantly expand the impact of this
approach.

\paragraph{Handling Unstructured Control Flow}
A natural next step is to extend SPL decomposition to accommodate
\texttt{goto} statements, exception handling, and other forms of unstructured
control flow. One approach is to augment the SPL grammar with special
constructs that capture irreducible control flow, or to treat such edges as
annotations handled separately by the dynamic program. Another possibility is
to combine SPL and tree decompositions: identify structured regions where SPL
applies and fall back to treewidth-based methods in irregular regions. Such a
hybrid decomposition may preserve the advantages of both worlds.

\paragraph{Applications Beyond Graph Coloring}
Several problems that do not resemble register allocation or LOSPRE 
may still benefit from small decomposition boundaries. For example, $\mu$-calculus model checking and related modal logics can be solved efficiently
on bounded-treewidth graphs. Whether SPL decomposition can offer similar
advantages---perhaps by using directionality or control-flow constructs to
improve symbolic propagation---is an interesting direction for research.
Investigating fixed-parameter and XP classifications for such problems under SPL decomposition may reveal new complexity-theoretic insights.

\paragraph{Integration into Compiler Pipelines}
Finally, applying SPL decomposition as a general optimization framework within existing compiler infrastructures (e.g., LLVM, GCC, or domain-specific compilers) may yield performance improvements across a range of analyses. Understanding how SPL interacts with SSA form, register coalescing, instruction scheduling, and other optimization passes is an important practical direction.

\medskip

Overall, our study points toward a broad landscape of opportunities for
advancing decomposition-based techniques in program analysis. The SPL
framework, by tightly coupling program structure with graph-theoretic
decomposition, offers a promising foundation for building scalable and principled
compiler optimizations.

\section{Conclusion}
This work introduces a novel grammar-based decomposition for control-flow graphs of structured programs, enhancing bottom-up dynamic programming while providing a computationally efficient alternative to traditional treewidth methods, significantly improving runtime in compiler optimization tasks. We applied this decomposition to improve the asymptotic runtimes of spill-free and minimum-cost register allocation, achieving exponential improvements related to the number of variables and registers. Additionally, we developed a lifetime-optimal speculative partial redundancy elimination (LOSPRE) algorithm for goto-free programs, which utilizes series-parallel-loop (SPL) decomposition to exploit graph sparsity, resulting in a remarkable 6 times runtime improvement over benchmarks from the Small Device C Compiler. Our method shows potential for significant runtime enhancements across various compiler optimization tasks, program verification, and model checking, although it is currently limited to goto-free programs. Future work will aim to extend the SPL decomposition to include goto statements and handle programs with exception handling or labeled loops.

\section*{Acknowledgments and Notes}

\noindent This work is an extension of the conference papers~\cite{ourasplos,oursetta}. The research was partially supported by the ERC Starting Grant 101222524 (SPES) and the Ethereum Foundation Research Grant FY24-1793. Chun Kit Lam was supported by a Hong Kong PhD Fellowship. Authors are ordered alphabetically.



\bibliographystyle{elsarticle-num-names} 
\bibliography{sections/refs}

\end{document}